\documentclass[11pt,a4paper]{amsart}
\usepackage{hyperref}

\usepackage{amssymb}
\usepackage{amsthm}
\usepackage{amsmath}
\usepackage{a4wide}
\usepackage[thinlines]{easytable}
\usepackage{graphicx}

\newtheorem{thm}{Theorem}[section]
\newtheorem{lem}[thm]{Lemma}
\newtheorem{prop}[thm]{Proposition}

\newtheorem{cor}[thm]{Corollary}

\begin{document}

 \def \d {{\rm d}}
 \newcommand{\boldk}{\mbox{\boldmath$k$}}                   
 \newcommand{\boldl}{\mbox{\boldmath$l$}}                   
 \newcommand{\boldm}{\mbox{\boldmath$m$}}                   

\newcommand{\todo}[1]{$\clubsuit$ {\tt #1} $\clubsuit$}

\newcommand{\R}{{\mathbb R}}
\renewcommand{\H}{{\mathcal H}}
\renewcommand{\a}{{\mathcal A}}

\newcommand{\Con}{\mathcal{C}}
\newcommand{\G}{\mathcal{G}}
\newcommand{\eps}{\varepsilon}
\newcommand{\D}{\ensuremath{{\mathcal D}}}
\newcommand{\rom}[1]{\text{\sf #1}}

\title{Completeness of general \emph{pp}-wave spacetimes and their impulsive 
limit}

\author[C. S\"amann]{Clemens S\"amann}
\address{Faculty of Mathematics, University of Vienna, Oskar-Morgenstern-Platz
1, 1090 Vienna, Austria}
\email{clemens.saemann@univie.ac.at}
\author[R. Steinbauer]{Roland Steinbauer}
\address{Faculty of Mathematics, University of Vienna, Oskar-Morgenstern-Platz
1, 1090 Vienna, Austria}
\email{roland.steinbauer@univie.ac.at}
\author[R.~\v{S}varc]{Robert \v{S}varc}
\address{Institute of Theoretical Physics, Faculty of Mathematics and Physics,
Charles University in Prague, V~Hole\v{s}ovi\v{c}k\'ach~2, 180~00 Praha 8, Czech
Republic and \\
Faculty of Mathematics, University of Vienna, 
Oskar-Morgenstern-Platz
1, 1090 Vienna, Austria}
\email{robert.svarc@mff.cuni.cz}

\begin{abstract}
We investigate geodesic completeness in the full family of \emph{pp}-wave or 
Brinkmann spacetimes in their extended as well as in their impulsive form. This 
class of geometries contains the recently studied gyratonic \emph{pp}-waves, 
modelling the exterior field of a spinning beam of null particles, 
as well as NPWs, which generalise classical \emph{pp}-waves by allowing for a 
general wave surface. The problem of geodesic completeness reduces to the 
question of completeness of trajectories on a Riemannian manifold under an 
external force field. Building upon respective recent results we derive 
completeness criteria in terms of the spatial asymptotics of the profile 
function in the  extended case.
In the impulsive case we use a fixed point argument to show that 
irrespective of the behaviour of the profile function \emph{all} geometries in 
the class are complete.   
\end{abstract}
 
\maketitle
\today
 
\bigskip\noindent
PACS class 04.20.Jb, 04.30.-w, 04.30.Nk, 04.30.Db, 02.30.Hq

\bigskip\noindent
MSC class 83C15, 83C35, 83C10, 34A36  



\section{Introduction}\label{sec:intro}
Since the seminal work of Penrose \cite{Pen:65a} singularities in general 
relativity are usually understood as the presence of incomplete 
causal geodesics, i.e., geodesics which cannot be extended to all values of 
their parameter. In this work we 
study geodesic completeness for a large class of spacetimes admitting a 
covariantly constant null vector field, forming the well known 
\emph{pp}-wave subclass of the \emph{Kundt} family of non-twisting, shear-free 
and non-expanding geometries \cite{Kun:61, Kun:62}. This subfamily includes, 
e.g.\ \emph{gyratonic pp-waves} \cite{FF:05,FZ:06}, representing the exterior 
vacuum field of spinning particles moving with the speed of light, which may 
serve as an interesting toy model in high energy physics \cite{YZF:07}, but 
also NP-waves, a generalisation of classical \emph{pp}-waves allowing for an 
$n$-dimensional Riemannian manifold $N$ as the wave surface \cite{CFS:03}. 
Remarkably this family of exact spacetimes has already been described by the 
original Brinkmann form \cite{Bri:25} of the \emph{pp}-wave metric, given in 
equation \eqref{bm}, below. 

We will consider geodesic completeness both in the extended case, i.e., where 
the profile functions are smooth, as well as in the impulsive case, i.e., when 
the metric functions are strongly concentrated and of a distributional nature.
The problem of completeness in this class of spacetimes can be reduced to a 
purely Riemannian problem, namely the question of completeness of the 
motion on a Riemannian manifold under the influence 
of a time and velocity dependent force. In the extended case we 
generalise recent results (\cite{CRS:12,CRS:13}) to the case at hand to provide 
completeness statements subject to conditions on the spatial fall off of the 
profile function. We also consider the case of impulsive waves in our class, 
which are models of short but violent bursts of gravitational radiation emitted 
by a spinning (beam of) ultrarelativistic particle(s). These models are also 
interesting from a purely mathematical point of view, since they are examples of 
geometries of low regularity, which attracted some growing attention recently, 
see e.g.\ \cite{CG:12, Min:15, Sbi:15, KSSV:15, Sae:16}. Our main result here 
is that \emph{all} impulsive geometries in the full class of \emph{pp}-waves 
are geodesically complete irrespective of the spatial asymptotics of the 
profile function. This result confirms the effect (previously noted in 
similar situations, cf.\ e.g.\ \cite{PV:99,SS:12}) that the influence the 
spatial characteristics of the profile function exert on the behaviour of the 
geodesics is wiped out in the impulsive limit. In this way we prove 
a large class of interesting non-smooth geometries to be \emph{non-singular} in 
view of the Penrose definition. 
\medskip

This article is structured as follows. In section \ref{sec:bm} we introduce the 
full \emph{pp}-wave or Brinkmann metric and summarize its basic geometric 
properties and its algebraic structure. Then, in section \ref{sec3} we deal 
with geodesic completeness in the extended case. We review previous 
works and apply recent results on the completeness of solutions to second 
order equations on Riemannian manifolds to the geodesic equations for the full 
Brinkmann metric. In particular, we establish completeness for a large class of
physically reasonable (extended) gyratonic  \emph{pp}-waves. In section 
\ref{sec:il} we introduce \emph{impulsive} geometries within our general class 
of \emph{pp}-wave solutions, thereby leaving the realm of classical smooth 
Lorentzian geometry. In section \ref{sec:cil} we consider the regularised 
version of these geometries and establish their completeness combining the 
results of section \ref{sec3} with a fixed point argument. Finally, in 
section \ref{sec:lim} we explicitly calculate the limits of the complete 
regularised geodesics given in section \ref{sec:cil} and relate them to 
the geodesics of the background spacetime. Physically this amounts to 
calculating the geodesics in the distributional model of the impulsive wave.

\section{The spacetime metric}\label{sec:bm}

In this section we introduce the \emph{full pp-wave} or \emph{Brinkmann 
metric} (\cite{Bri:25}) and review some of its basic geometric properties along 
with relevant subclasses and special cases which have been treated extensively 
in 
the literature. 

To begin with let
$(N,h)$ be a smooth connected $n$-dimensional Riemannian manifold. We consider 
the spacetime $(M=N\times\R^2,g)$  where the line element is given by
\begin{equation}\label{bm}
 ds^2=h_{ij}dx^idx^j-2dudr+\H(x,u)du^2+2\a_i(x,u)dudx^i \,.
\end{equation}
Here ${x=x^i=(x^1,\dots,x^n)}$ are coordinates on $N$ 
and $u,r$ are global coordinates on $\R^2$. Moreover $\H:\ N\times\R\to\R$ and
$\a_i:\ N\times\R\to\R$ are smooth functions. We fix a time orientation on 
$M$ by defining the null vector field $\partial_r$ to be future directed. 

Some immediate geometric properties of the spacetime \eqref{bm} are the 
following: 
$\partial_r$ is the generator of the null hypersurfaces  of constant 
$u$, $P(u_0):=\{u=u_0\}\cong N\times\R$, i.e., 
$\partial_r=-\mathrm{grad}u=-\nabla u$ and it is covariantly 
constant. The latter property is the defining condition for \emph{pp}-waves
(plane fronted waves with parallel rays), and this is why we will refer to 
\eqref{bm} as full \emph{pp}-wave, cf.\ e.g.\ \cite[p.\ 324 and Sec.\ 18.5]{GP:09}.

The null geodesic generators of $P(u_0)$ form a non-expanding,
shear-free and twist-free congruence and the family of $n$-dimensional 
spacelike submanifolds $N\times\{u_0\}$ orthogonal to $\partial_r$ has the 
interpretation of wave surfaces. Consequently we will also refer to $N$ as the 
wave surface, and to $h$ and $x^i$ as the spatial metric and coordinates, 
respectively. 

Using the coordinates ${(x^i, u, r)}$ the inverse metric takes the form
\begin{equation} \newcommand*{\tempb}{\multicolumn{1}{|c}{B}}
 g^{\mu\nu}=\begin{pmatrix}
  h^{ij}&0&g^{ri}\\
  0&0&-1\\
  g^{ri}&-1&g^{rr}
 \end{pmatrix}\,,
\end{equation}
where $h^{ij}$ denotes the components of the inverse
spatial metric $h^{-1}$ on $N$, and we have $g^{rr}=-\H+h^{ij}\a_i\a_j$ and
$g^{ri}=h^{ij}\a_j$. The non-vanishing Christoffel symbols then are
\begin{align}
\Gamma^{r}_{uu}&=g^{ri}(\a_{i,u}-\tfrac{1}{2}\H_{,i})-\tfrac{1}{2}{\H_{,u}} \,,
\nonumber\\
 \Gamma^{r}_{uj}&=\tfrac{1}{2}g^{ri}(\a_{i,j}-\a_{j,i})-\tfrac{1}{2}\H_{,j} \,, \quad
 \Gamma^{r}_{jk}= -\a_{(j||k)}\,, \label{ChrSymb}\\
 \Gamma^{i}_{uu}&=h^{ik}(\a_{k,u}-\tfrac{1}{2}\H_{,k}) \,, \quad 
 \Gamma^i_{uj}=\tfrac{1}{2}h^{ik}(\a_{k,j}-\a_{j,k}) \,,\quad
 \Gamma^{i}_{jk}={\Gamma^{(N)}}^i_{jk} \,,\nonumber
\end{align}
where $\Gamma^{(N)}$ denotes the Christoffel symbols of the Riemannian metric
$h$ on $N$, $\,_{||}$ stands for the covariant derivative of $h$, and 
${\,_{,i}}$  and $\,_{,u}$  denote derivatives with respect to the
$i$-th spatial direction and with respect to $u$, respectively. 

From the vanishing of all Christoffel symbols of the form $\Gamma^u_{\mu\nu}$ 
we immediately find that for any geodesic $\gamma(s)=(x^i(s),u(s),r(s))$ in 
\eqref{bm} we have $\ddot u=0$.
Hence there are geodesics with $u(s)=u_0$ that are entirely 
contained in the null hypersurface $P(u_0)$ and thus are either spacelike or 
the null generators of $P(u_0)$. Observe that in case $\a_i=0$, also the 
$r$-component of these geodesics becomes affine, $r(s)=r_0+\dot 
r_0s$. All other geodesics may be rescaled to take the form 
\begin{align}\label{ng}
 \gamma(s)=(x^i(s),s,r(s)) \,.
\end{align}

The coordinate function $u$ is increasing along any future directed causal 
curve $c(t)=(x^i(t),u(t),r(t))$ since
$(u\circ c)\dot{}=g(\nabla u,\dot c)=\dot u \geq 0$
and it is even strictly increasing along any future directed timelike 
curve. Hence $(M,g)$ is chronological. Moreover $u$ is strictly increasing 
along all future directed causal geodesics of the form \eqref{ng}.
So in case $\a_i=0$ there is no closed null geodesic segment and the spacetime 
is even causal.
 
The non-vanishing components of the Ricci tensor for \eqref{bm} are
\begin{equation}\label{GeneralRicci}
R_{ij}=R_{ij}^{(N)} \,, \quad R_{ui}=h^{mn}\a_{[m,i]||n} \,, \quad 
R_{uu}=h^{mn}(\a_{m,u||n}-\tfrac{1}{2}\H_{||mn})+h^{kl}h^{mn}\a_{[k,m]}\a_{[l,n]
} \,,
\end{equation}
where the square brackets, as usual, denote antisymmetrisation.
The Ricci scalar of \eqref{bm} corresponds to that of the 
transverse space, i.e., ${R=R^{(N)}}$. Also, the metric \eqref{bm} belongs to 
the class of  VSI (vanishing scalar invariant) spacetimes iff $N$ is flat.

Next we employ the algebraic Petrov classification (\cite{OPP:13,PS:13}) 
to the \emph{pp}-wave metric (\ref{bm}). We project the Weyl tensor onto the  
natural null frame 
\begin{equation}
\boldk=\mathbf{\partial}_r \,, \ \quad 
\boldl=\tfrac{1}{2}\,\H\,\mathbf{\partial}_r+\mathbf{\partial}_u \,, \ \quad 
\boldm_{(i)}=m_{(i)}^i\,(\a_{i}\,{\partial}_r+\mathbf{\partial}_i) \,, 
\label{nat null frame}
\end{equation}
where ${h_{ij}\,m_{(i)}^im_{(j)}^j=\delta_{ij}}$ and find that the highest 
boost weight irreducible components ${\Psi_{0^{ij}}}$, ${\Psi_{1T^{i}}}$ and 
${\tilde{\Psi}_{1^{ijk}}}$ vanish and the Brinkmann spacetimes are thus 
necessarily at least of algebraic type II with ${\boldk=\mathbf{\partial}_r}$ 
being a double degenerated null direction, in fact ${\Psi_{2^{ij}}=0}$ and 
(\ref{bm}) is of type II(d). More precisely, without employing any field 
equations the non-vanishing Weyl scalars are
\begin{align}
\Psi_{2S} &= \tfrac{1}{n(n+1)}\,R^{(N)}\,, \label{Psi2s} \\
\tilde{\Psi}_{2T^{(ij)}} &= 
\tfrac{1}{n}\,m_{(i)}^im_{(j)}^j\,\big(R^{(N)}_{ij}-\tfrac{1}{n}\,h_{ij}\,R^{(N)
} \big) \,, \label{Psi2Tij} \\
\tilde{\Psi}_{2^{ijkl}} &= m_{(i)}^i m_{(j)}^j m_{(k)}^k 
m_{(l)}^l\,C^{(N)}_{ijkl} \,, \label{Psi2ijkl}\\
\Psi_{3T^{i}} &= \tfrac{1}{n}\,m_{(i)}^i\,h^{kl}\a_{[i,k]||l} \,, 
\label{Psi3Tj}\\
\tilde{\Psi}_{3^{ijk}} &= 
m_{(i)}^im_{(j)}^jm_{(k)}^k\,\Big(\a_{[k,j]||i}-\tfrac{1}{n-1}\,h^{mn}\,\big(h_{
ij}\a_{[k,m]||n}-h_{ik}\a_{[j,m]||n}\big)\Big) \,, \label{Psi3ijk}\\
\Psi_{4^{ij}} &= 
m_{(i)}^im_{(j)}^j\,\Big(-\tfrac{1}{2}\H_{||ij}+\a_{(i,u||j)}+h^{mn}\a_{[m,i]}
\a_{[n,j]} \nonumber \\
& \hspace{25.0mm} -\tfrac{1}{n}\,h_{ij} 
h^{kl}\big(-\tfrac{1}{2}\H_{||kl}+\a_{k,u||l}+h^{mn}\a_{[m,k]}\a_{[n,l]}
\big)\Big) \,. \label{Psi4ij}
\end{align}
The boost weight zero components 
${\Psi_{2S},\,\tilde{\Psi}_{2T^{(ij)}},\,\tilde{\Psi}_{2^{ijkl}}}$ are entirely 
governed by the properties of the transverse space $N$, the components of boost 
weight $-1$, ${\Psi_{3T^{i}}}$, and ${\tilde{\Psi}_{3^{ijk}}}$, are governed by 
the off-diagonal terms $\a_i$, while the boost weight $-2$ component 
${\Psi_{4^{ij}}}$ depends on all the metric functions $\H$ and $\a_i$. The 
conditions under which the geometry (\ref{bm}) becomes algebraically more 
special are summarized in table~\ref{BrinkAlg_NoFE}. 

Moreover, if the vacuum Einstein field equations \eqref{GeneralRicci} are 
employed some of the conditions are satisfied identically, see 
table~\ref{BrinkAlg_EFE}.

{\small 
\renewcommand{\arraystretch}{1.5}
\begin{table}[h]
\begin{tabular}{l|l}
{\small (sub)type}\qquad&\quad {\small condition}  \\ \hline 
$\mathrm{II(a)}$ &\quad $\displaystyle{R^{(N)}=0}$ \\
$\mathrm{II(b)}$ &\quad ${R^{(N)}_{ij}=\tfrac{1}{n}\,h_{ij}\,R^{(N)}}$ \\
$\mathrm{II(c)}$ &\quad ${C^{(N)}_{ijkl}=0}$ \\
$\mathrm{II(d)}$ &\quad {\small always} \\ \hline
$\mathrm{III}$ &\quad $\mathrm{II(abcd)}$ \\ \hline
$\mathrm{III(a)}$ &\quad ${h^{kl}\a_{[i,k]||l}=0}$ \\
$\mathrm{III(b)}$ &\quad 
${\a_{[k,j]||i}=\tfrac{1}{n-1}\,h^{mn}\,\big(h_{ij}\a_{[k,m]||n}-h_{ik}\a_{[j,m]
||n}\big)}$ \\ \hline
$\mathrm{N}$ &\quad $\mathrm{III(ab)}$ \\ \hline
$\mathrm{O}$ &\quad 
$-\tfrac{1}{2}\H_{||ij}+\a_{(i,u||j)}+h^{mn}\a_{[m,i]}\a_{[n,j]}$ \\
&\quad \qquad $=\tfrac{1}{n}\,h_{ij} 
h^{kl}\big(-\tfrac{1}{2}\H_{||kl}+\a_{k,u||l}+h^{mn}\a_{[m,k]}\a_{[n,l]}\big)$ 
\\
\end{tabular}
\vspace*{12pt}
\caption{\label{BrinkAlg_NoFE}{\small  The conditions refining possible 
algebraic 
(sub)types of the Brinkmann geometry (\ref{bm}). No field equations 
are employed. The relations II(b), III(b) are identities for $n=2$ and II(c) 
for $n=2,3$, respectively. When II(abc) are satisfied simultaneously the 
transverse space has to be flat.}}
\end{table}
}

{\small 
\renewcommand{\arraystretch}{1.5}
\begin{table}[h]
\begin{tabular}{l|l}
{\small (sub)type}\qquad&\quad {\small condition}  \\ \hline
$\mathrm{II(a)}$ &\quad {\small always} \\
$\mathrm{II(b)}$ &\quad {\small always} \\
$\mathrm{II(c)}$ &\quad ${C^{(N)}_{ijkl}=0}$ \\
$\mathrm{II(d)}$ &\quad {\small always} \\ \hline
$\mathrm{III}$ &\quad $\mathrm{II(abcd)}$ \\ \hline
$\mathrm{III(a)}$ &\quad {\small always} \\
$\mathrm{III(b)}$ &\quad ${\a_{[k,j]||i}=0}$ \\ \hline
$\mathrm{N}$ &\quad $\mathrm{III(ab)}$ \\ \hline
$\mathrm{O}$ &\quad 
$-\tfrac{1}{2}\H_{||ij}+\a_{(i,u||j)}+h^{mn}\a_{[m,i]}\a_{[n,j]}=0$ \\
\end{tabular}\vspace*{12pt}
\caption{\label{BrinkAlg_EFE}{\small The algebraic structure of the spacetimes 
(\ref{bm}) restricted by Einstein's vacuum field equations. The algebraic type 
becomes in general II(abd). Moreover, the condition II(c) is identically 
satisfied in $n=2,3$ and III(b) in $n=2$, respectively. Condition II(c) 
necessarily leads to a flat transverse space.}}
\end{table} 
}

\bigskip 

The spacetimes \eqref{bm} have been originally considered by Brinkmann in the 
context of conformal mappings of Einstein spaces. Over the decades, starting 
with  Peres \cite{P:59}, several special cases of \eqref{bm} have been used as 
exact models of spacetimes with gravitational waves.
\medskip

The most prominent case arises 
when the wave surface $N$ is flat $\R^2$ and the off-diagonal terms $\a_i$ 
vanish. In fact, writing $x^i=(x,y)$ these `classical' \emph{pp-waves}, 
\begin{equation}\label{pp}
 ds^2=dx^2+dy^2-2dudr+\H(x,y,u)du^2\,,
\end{equation}
have become text book examples of exact gravitational wave spacetimes, see 
e.g.\ \cite[Ch.~17]{GP:09}. The Petrov type is N (cf.\ table~\ref{BrinkAlg_NoFE}) without the application of the field 
equations hence in any theory of 
gravity. Moreover, the Ricci tensor simplifies to $R_{uu}=-(1/2)\delta^{ij}\H_{,ij}$ 
(cf. \eqref{GeneralRicci}) and its source can be interpreted as any type of 
null matter or radiation. The vacuum Einstein equations reduce to the 
$2$-dimensional (flat) Laplacian so that \eqref{pp} with harmonic $\H$ 
represents pure gravitational waves. Such solutions are most conveniently 
written using the complex coordinate $\zeta=x+iy$ with 
$\H(\zeta,u)=F(\zeta,u)+ F(\bar\zeta,u)$ and $F$ a combination of terms of 
the form
\begin{equation}\label{ppp}
 F(\zeta,u) = \sum_{m=1}^\infty \alpha_m(u)\,\zeta^{-m} -\mu(u)\log\zeta + 
\sum_{m=2}^\infty \beta_m(u)\,\zeta^m\,
\end{equation}
with arbitrary `profile functions' $\alpha_m$, $\beta_m$, $\mu$ of $u$.
Here the inverse-power terms represent \emph{pp}-waves generated by sources with 
multipole structure 
(see e.g.\ \cite{PG:98}) moving along the axis  
which is clearly singular. Hence it has to be removed from the spacetime which 
now has spatial part $N=\R^2\setminus\{0\}$. The same is true for the extended 
Aichelburg--Sexl \cite{AS:71} solution represented by the logarithmic term. 
Finally, the polynomial terms are non-singular 
with $\beta_2$ representing plane waves, see below. The higher order terms 
($m\geq 3$) lead to unbounded curvature at infinity and display chaotic 
behaviour of geodesics (e.g.\ \cite{PV:98}) hence seem to be physically less 
relevant. 

In the special case of $\H$ being quadratic in $x,y$ one arrives at \emph{plane 
waves}, i.e.,
\begin{align}\label{pw}
 ds^2=dx^2+dy^2-2dudr+\H_{ij}(u)x^ix^jdu^2 \,,
\end{align}
where $\H_{ij}$ is a real symmetric $(2\times2)$-matrix-valued function on $\R$.
Here the curvature tensor components are constant along the wave surfaces and
the spacetime is Ricci-flat provided the trace of the profile function 
vanishes, $\H^i_i=0$, in which case we speak of (purely) gravitational 
plane waves. While being complete by virtue of the linearity 
of the geodesic equations, 
plane waves `remarkably' fail to be globally hyperbolic 
(\cite{Pen:65b}) due to a focussing effect of the null geodesics. This phenomenon 
has been investigated thoroughly for gravitational plane waves in a series of 
papers by Ehrlich and Emch (\cite{EE:92a,EE:92b,EE:93}, see also 
\cite[Ch.\ 13]{BEE:96}), who were able to determine their precise position on 
the 
causal ladder: They are causally continuous but not causally simple. In this 
analysis, however, the high degree of symmetries of the flat wave surface was 
extensively used and the `stability' of the respective properties of plane 
waves within larger classes of solutions remained obscure.
\medskip

Partly to clarify these matters Flores and S{\'a}nchez, in part together with 
Candela, in a series of papers (\cite{CFS:03,FS:03,CFS:04,FS:06}) introduced 
more general models which they called (general) plane-fronted waves (PFW).
Indeed they generalise \emph{pp}-waves \eqref{pp} by replacing the flat 
two-dimensional wave surface by an arbitrary $n$-dimensional Riemannian 
manifold $N$, i.e.,
they are defined as \eqref{bm} with vanishing off-diagonal terms
$\a_i$,
\begin{align}\label{npw}
 ds^2=h_{ij}dx^idx^j-2dudr+\H(x,u)du^2 \,.
\end{align}
Motivated by the geometric interpretation given above 
and following \cite{SS:12} we call these models \emph{$N$-fronted 
waves with parallel rays (NPWs)}. The Petrov type of these geometries
is now II(d) (cf.\ table \ref{BrinkAlg_NoFE}) and the 
vacuum field equations become $\Delta_h \H=h^{ij}\H_{||ij}=0$, i.e., the 
Laplace equation for $\H$ on $(N,h)$, which in addition has to be Ricci-flat. 
Hence vacuum NPWs
are necessarily of Petrov type II(abd) and are of type $N$ if and only if 
$(N,h)$ in addition is conformally flat, hence flat.
It turns out that the behaviour of $\H$ at 
spatial infinity is decisive for many of the global properties 
of NPWs with quadratic behaviour marking the critical case: 
NPWs are causal but not necessarily distinguishing, they are
strongly causal if $-\H$ behaves at most quadratically at spatial 
infinity\footnote{For precise definitions of these conditions see 
Section \ref{sec3}, below.} and
they are globally hyperbolic if $-\H$ is subquadratic and $N$ is complete.
Similarly the global behaviour of geodesics in NPWs is governed by the behaviour
of $\H$ at spatial infinity. The respective results will be discussed in the 
next section together with the stronger and newer results 
of \cite{CRS:12,CRS:13}.

\medskip
A complementary generalisation of \emph{pp}-waves (\ref{pp}) was considered 
by Bonnor (\cite{Bonnor:1970b}) and independently by Frolov and his 
collaborators in \cite{FF:05,FIZ:05,FZ:06,YZF:07}. Here the $n$-dimensional 
transverse space $N$ is considered to remain flat but non-trivial 
off-diagonal terms $\a_i(x,u)$ are allowed to obtain \emph{gyratonic 
pp-waves}
\begin{equation}\label{pp-gyraton}
 ds^2=\delta_{ij}dx^idx^j-2dudr+\H(x,u)du^2+2\a_i(x,u)dudx^i \,.
\end{equation}
Physically this geometry represents a \emph{spinning null beam of pure 
radiation}, first called \emph{gyraton} in \cite{FF:05}.
By flatness of the wave surface the Petrov type is at least III and in case of 
vacuum solutions at least III(a) and N if $n=2$ (cf.\ tables 
\ref{BrinkAlg_NoFE}, \ref{BrinkAlg_EFE}). 
In \cite{Bonnor:1970b} the metric \eqref{pp-gyraton} was matched 
to an `interior' non-vacuum region, where the spinning source of the 
gravitational waves was given phenomenologically by an energy momentum tensor 
of the form $T_{uu}=\rho$ and $T_{ui}=j_i$. In the surrounding vacuum region 
the metric functions $\H$ and $\a_i$ 
are restricted by the vacuum Einstein equations following from 
(\ref{GeneralRicci}) for ${R_{ui}=0}$ and ${R_{uu}=0}$,
\begin{equation}
\delta^{mn}\a_{[m,i],n}=0 \,, \qquad 
-\tfrac{1}{2}\,\delta^{mn}\H_{,mn}+(\delta^{mn}\a_{m,n})_{,u}+\delta^{kl}
\delta^{mn}\a_{[k,m]}\a_{[l,n]}=0 \,,
\end{equation}
where in addition the `Lorenz' gauge ${\delta^{mn}\a_{m,n}=0}$ can be applied. 
In \cite{FF:05,FIZ:05} explicit solutions have been calculated, all displaying 
a fall-off like inverse powers of $\|x\|$ if $n>2$ and a logarithmic 
behaviour if $n=2$. \label{FF-asymptotics}
Moreover, the weak field approximation in the presence of a gyratonic 
source, reflected in the only non-trivial energy-momentum tensor components 
$T_{uu}$ and $T_{ui}$, gives physical meaning to the metric functions in the 
entire spacetime. In particular, $\H$ determines the mass-energy density and 
$\a_i$ correspond to the angular momentum density of the source. 

In four dimensions (${n=2}$) the off-diagonal terms in \eqref{pp-gyraton} can 
be removed \emph{locally} by a coordinate transformation which, however, 
obscures the global (topological) properties of the spacetime, cf.\ e.g.\ 
\cite{FIZ:05,PSS:14}. In higher dimensions (${n>2}$), the off-diagonal terms in 
\eqref{pp-gyraton} can anyway not be removed even locally due to the lack of 
coordinate freedom. 

The spinning character of four-dimensional gyratonic \emph{pp}-wave 
(\ref{pp-gyraton}) was emphasized in \cite{FIZ:05,PSS:14} employing transverse 
polar coordinates ${x=\rho\cos\varphi}$, ${y=\rho\sin\varphi}$ with the 
identifications ${\a_1=-J\rho^{-1}\sin\varphi}$, ${\a_2=J\rho^{-1}\cos\varphi}$ 
which gives \begin{equation} ds^2=d\rho^2+\rho^2 d\varphi^2-2dudr+\H(\rho, 
\varphi, u)du^2+2J(\rho,\varphi,u)dud\varphi \,. \label{ppGyrPolar} 
\end{equation} This can also be understood as the specific form of the full 
Brinkmann metric (\ref{bm}) with ${x^1=\rho}$, ${x^2=\varphi}$ and 
${dh^2=d\rho^2+\rho^2 d\varphi^2}$. An important explicit example is given by 
the axially symmetric gyraton accompanied by a \emph{pp}-wave generated by a 
monopole, namely 
\begin{equation}\label{pp-gyraton_axysym_4d} 
ds^2=d\rho^2+\rho^2 d\varphi^2-2dudr-2\mu(u)\ln\rho du^2+\chi(u)dud\varphi \,, 
\end{equation} 
where $\mu(u)$ determines the mass density in the Aichelburg-Sexl-like 
logarithmic term, and the angular momentum density is determined by $\chi(u)$.

\section{Geodesic completeness}\label{sec3}

In this section we discuss geodesic completeness of the full \emph{pp}-wave
metric \eqref{bm} and its various special cases in the extended, i.e., smooth 
case before turning to the impulsive limit in subsequent sections.  We 
will see that all these 
questions can be reduced to a purely Riemannian question about 
the completeness of trajectories on $N$ under a specific force field.

The explicit form of the system of geodesic equations
for a curve $\gamma(s)=(x^i(s),u(s),r(s))$ is, see (\ref{ChrSymb}),
\begin{align}\label{ge}
  \ddot x^i&=-{\Gamma^{(N)}}^i_{jk}\,\dot x^j\dot
   x^k-h^{ik}(\a_{k,j}-\a_{j,k})\,\dot x^j \dot u
             -\tfrac{1}{2}h^{ik}(2\a_{k,u}-\H_{,k})\,\dot u^2 \,, \\
 \ddot u&=0 \,, \\
 \ddot r&=
 \a_{(i||j)}\,\dot x^i\dot x^j-\Big(g^{ri}(\a_{i,j}-\a_{j,i})-\H_{,j}\Big)\,\dot u\dot x^j
  -\Big(g^{ri}(\a_{i,u}-\tfrac{1}{2}\H_{,i})-\tfrac{1}{2}\H_{,u}\Big)\,\dot u^2 
   \,.
\end{align}

We observe (again) that the equation for $u$ is trivial
and that the equation for $r$ decouples from the rest of the system and can 
simply be integrated once the $x$-equations are solved. Finally the 
$x$-equations are the equations of motion on the Riemannian manifold 
$N$ under an external force term depending on time and velocity. Hence the 
basic result on the form of the geodesics of the Brinkmann metric (\ref{bm}) is 
the following (cf.\ \cite[Prop.\ 3.1]{CFS:03}):

\begin{prop}[Form of the geodesics]\label{prop:fog} 
 Let $\gamma=(x^i,u,r):\, (-a,a)\to M$ be a curve on $M$ with constant 
 `energy' $E_\gamma=g(\dot \gamma,\dot\gamma)$ which assumes the data
 \begin{align}
   \gamma(0)
     =(x^i_0,u_0,r_0) \,,\quad
   \dot\gamma(0)
    =(\dot x^i_0,\dot u_0,\dot r_0) \,.
 \end{align}
 Then $\gamma$ is a geodesic iff the following conditions hold true
 \begin{enumerate}
  \item [(a)] $u$ is affine, i.e., $u(s)=u_0+s\dot u_0$ for all $s\in(-a,a)$,
  \item [(b)] $x^i$ solves 
   \begin{equation}\label{req}
    D_{\dot x}\dot x^i=-h^{ik}(\a_{k,j}-\a_{j,k})\,\dot x^j \dot u_0
             -\tfrac{1}{2}h^{ik}(2\a_{k,u}-\H_{,k})\dot u_0^2 \,,
   \end{equation}
  where $D$ denotes the covariant derivative of $(N,h)$,
  \item[(c)] $r$ is given by
   \begin{align}\label{eq:rgeo}
    r(s)=&\,r_0\nonumber\\
         &-\frac{1}{2\dot u_0}\int\limits_{0}^{s}
     \left( E_\gamma-h(\dot x(\sigma),\dot x(\sigma))
            -\dot u_0^2H(x(\sigma),u(\sigma))
            -2\dot u_0\a_i(x(\sigma),u(\sigma))\dot 
x^i(\sigma)\right)\d\sigma\,.
   \end{align}
 \end{enumerate}
\end{prop}

Here we have used the explicit form of $E_\gamma$ in the last condition. The 
key fact is now that completeness essentially depends on completeness of the 
solutions to \eqref{req}. Note that in most cases we will use a rescaling to 
achieve the form \eqref{ng} of the geodesics, which amounts to setting $u_0=0$ 
and $\dot u_0=1$ in equations \eqref{req} and \eqref{eq:rgeo}.

\begin{cor}[Basic condition for completeness]
 The spacetime \eqref{bm} is complete iff 
 all inextendible solutions of \eqref{req}
 are complete.
\end{cor}

Although there are some results also in case of an incomplete spatial manifold, 
see the discussion after Corollary \ref{NPWc} below, we assume for the moment 
$(N,h)$ 
to be complete.
The question of completeness of solutions of equations as \eqref{req} has,
if only in special cases, been addressed in the `classical' literature. To begin 
with, we observe that in the special case of NPWs \eqref{npw} the equation of 
motion \eqref{req} reduces to 
\begin{equation}\label{rpot}
 D_{\dot x}\dot x=\frac{1}{2}\nabla_x \H(x,s)\,
\end{equation}
($\nabla_x$ denoting the gradient on $(N,h)$), i.e., to the equation of 
motion on $N$ under the influence of a time dependent potential. First 
results on the completeness of NPWs, mainly restricted to the case of 
autonomous $\H$, i.e., $\H=\H(x)$ independent of $u$ were derived in 
\cite[Sec.~3]{CFS:03}. 
In fact, it follows from e.g.\ \cite[Thm.\ 3.7.15]{AM78} that 
a NPW \eqref{npw} is complete if $\H$ is autonomous and controlled by a 
\emph{positively complete} function at infinity, i.e., if there exists
some arbitrarily fixed $\bar x \in N$ and some positive constant $\mathfrak{R}$ 
such that 
\begin{align}
 \H=\H(x)\leq -V(d(x,\bar x))\quad \text{for all $x\in N$ with $d(x,\bar x)\geq 
\mathfrak{R}$}\,,
\end{align}
where $d$ denotes the Riemannian distance function on $(N,h)$ and $V:\, 
[0,\infty)\to\R$ is $C^2$ with 
$\int_0^\infty dx/\sqrt{e-V(x)}=+\infty$ for one (hence any) $e>V(x)$ and all 
$x$. Consequently autonomous NPW are complete if $\H$ \emph{grows at most 
quadratically at spatial infinity}, i.e, if $\exists\ \bar x\in 
N,\ \mathfrak{R}>0$ such that 
\begin{align} 
 \H=\H(x)\leq C\, d^2(x,\bar x)\quad\text{for all $x$ with $d(x,\bar x)\geq 
\mathfrak{R}$}\,,
\end{align}
for some constant $C>0$. Of course this result extends immediately to 
\emph{sandwich}\footnote{We call a spacetime \eqref{bm} a sandwich wave if $\H$ 
and $\a_i$ vanish outside some bounded $u$-interval.} NPWs, which grow at most 
quadratically at spatial infinity. Also the case of \emph{plane} NPWs is easily 
settled, that is 
\eqref{npw} with $(N,h)$ flat and quadratic non-autonomous $\H$, i.e.,
\begin{align}
 \H(x,u)=h(A(u)x,x)\,,
\end{align}
where $A$ is (at least) a continuous map from $\R$ into the space of real 
symmetric $(n\times n)$-matrices. Here completeness follows 
from global existence of solutions to linear ODEs generalising the case of 
plane waves \eqref{pw}.

More substantial results on non-autonomous NPWs have been given in \cite{CRS:12}
based on recent results on the completeness of trajectories of equations like
\eqref{req} and \eqref{rpot} in \cite{CRS:13}. (For more general and 
somewhat sharper results see \cite{Min:15}.) Since we will use these 
statements also in our discussion of the general case \eqref{bm} we recall in 
the following the key notions and theorems. We say that a (time dependent) 
tensor field $X$ on 
the projection $\pi:N\times \R\to N$ \emph{grows at most linearly in $N$ along 
finite times} if for all $T>0$ there exists $\bar x\in N$ and constants 
$A_T,C_T>0$
such that 
\begin{align}\label{aft}
 |X|_{(x,s)}\leq A_T\, d(x,\bar x)+C_T\quad \forall (x,s)\in N\times[-T,T]
\end{align}
with $|\ |$ and $d$ the norm and the distance function of $h$, respectively.
Analogously we define the notions of \emph{at most quadratic growth along finite 
times} and \emph{boundedness along finite times}, where in the 
special case of functions we use the estimate \eqref{aft} without norm.
Now given a smooth $(1,1)$-tensor field $F$ and a smooth vector field
$X$ on $\pi$
we consider the second order ODE
\begin{align}\label{ODE}
D_{\dot\gamma}\dot\gamma(s)&=F_{(\gamma(s),s)}\dot\gamma(s)+X_{(\gamma(s),s)}\,
\end{align}
and the special case when $X$ is derived from a potential, i.e.,
\begin{align}\label{ODEpot}
D_{\dot\gamma}\dot\gamma(s)&=F_{(\gamma(s),s)}\dot\gamma(s)-\nabla_x 
V{(\gamma(s),s)}\,
\end{align}
with $V$ a smooth function on $N\times\R$. Then we have

\begin{thm}[Theorems 1,2 in \cite{CRS:13}]\label{crs-thm} 
Let $(N,h)$ be a connected, complete Riemannian manifold. If the self 
adjoint part $S$ of $F$ is bounded in $N$ along finite times then
\begin{enumerate}
 \item all inextendible solutions of \eqref{ODE} are complete provided $X$ 
grows at most linearly in $N$ along finite times, and
 \item all inextendible solutions of \eqref{ODEpot} are complete provided that
  $-V$ and $|\frac{\partial V}{\partial s}|$ grow at most quadratically in $N$ 
along finite times. 
\end{enumerate}
\end{thm}

Observe that one may also apply Theorem \ref{crs-thm}(1) to equation 
\eqref{ODEpot} in which case one has to assume that $\nabla_xV$ grows at most 
linearly along finite times. Provided that we are in the non-autonomous 
case this condition is logically independent of the condition of 
Theorem~\ref{crs-thm}(2). 

Hence in the case of NPWs, which amounts to setting $F=0$ and
$X=-\nabla_xV=\nabla_x\H$, one obtains different types of results based on 
either of these conditions, see \cite{CRS:12,CRS:13}. Explicitly we 
have

\begin{cor}[Completeness of NPWs and classical \emph{pp}-waves]\label{NPWc}
NPW spacetimes \eqref{npw} and, in particular, classical \emph{pp}-wave 
spacetimes \eqref{pp} with complete wave surface $N$ are complete provided that 
either
\begin{enumerate}
 \item $\nabla_x\H$ grows at most linearly along finite times, or 
 \item $\H$ and $\left|\frac{\partial \H}{\partial s}\right|$ grow at most 
  quadratically along finite times,  or
 \item $\H(x,u)\leq\beta_0(u)$ and 
  $\left|\frac{\partial\H}{\partial u}(x,u)\right|\leq 
\alpha_0(u)\big(\beta_0(u)-\H(x,u)\big)$ for some continuous real functions 
$\alpha_0$, $\beta_0$ and all $(x,u)\in N\times\R$.
\end{enumerate}
\end{cor}

Condition (3) is, however, not derived from Theorem \ref{crs-thm} but due to
\cite[Cor.\ 3.3]{CRS:12} and again logically independent of the other 
conditions. A physically interesting consequence of condition (2), which 
actually generalises the above results on autonomous and sandwich NPWs of 
quadratic growth, is that it provides stability of completeness of plane waves 
within the class of NPWs with quadratic behaviour of $\H$, cf.\ \cite[Rem.\ 
3.5]{CRS:12}.

Observe that physically reasonable models of classical gravitational 
\emph{pp}-waves, as 
discussed below equation \eqref{ppp}, possess a non-complete wave surface and 
hence Corollary \ref{NPWc} does not apply in this case. However, the 
geodesics will still be `complete at infinity' since the asymptotic 
conditions of Corollary \ref{NPWc} hold true for the multipole as well as for 
the logarithmic terms in \eqref{ppp}. However, the geodesics could leave the 
exterior region `at the inside' proceeding to the matter region. This behaviour 
clearly has to be considered as physically reasonable. Also mathematically 
completeness of trajectories of \eqref{ODE}, \eqref{ODEpot} on 
\emph{incomplete} Riemannian manifolds is subject to 
very strong conditions, see e.g. \cite{Gordon}: A sufficient condition, e.g. 
is that $\H$ is proper and bounded from below, which certainly does not hold in 
our case. Note that this applies to wave surfaces of the form 
$N=\R^n\setminus\{0\}$ as well as to those of the form $\R^n$ with a (closed) 
ball removed. In the latter case one would of course match the solution to some 
non-vacuum interior region inside the ball. The situation is of course 
completely analogous in case of NPWs.

\medskip

Turning now to the general case, i.e., to the quest for completeness of the 
full \emph{pp}-wave geometry \eqref{bm} we more extensively make use of the
power of Theorem \ref{crs-thm}. Indeed  
$F^i_j=-h^{ik}(\a_{k,j}-\a_{j,k})$ is no longer vanishing but 
still its selfadjoint part satisfies $S=0$ so that Theorem \ref{crs-thm} puts 
no restriction on $\a_{k,j}$. On the other hand, 
$X=-\frac{1}{2}h^{ik}(2\a_{k,u}-\H_{,k})$ and
we can no longer write $X$ as the gradient of a potential. So we cannot use 
condition (2) and have to exclusively resort to Theorem \ref{crs-thm}(1).
In this way we obtain 

\begin{cor}[Completeness of the Brinkmann metric]\label{cor:cbm}
The full \emph{pp}-wave spacetime \eqref{bm} is complete if $N$ is 
complete and $\nabla_x\H$ and $h^{ik}\a_{k,u}$ grow at most linearly along 
finite times. 
\end{cor}

Finally, we come to discuss completeness of gyratonic \emph{pp}-waves 
\eqref{pp-gyraton}.  In this case the wave surface is flat 
and so we only have to deal with the asymptotics of the metric 
functions. 
\begin{cor}[Completeness of gyratons]
 Any gyratonic \emph{pp}-wave \eqref{pp-gyraton} with $N=\R^n$ and $\H_{,x}$ as well as
 $\a_{k,u}$ growing at most linearly along finite times is complete.  
\end{cor}

Now the asymptotics of the explicit gyratonic \emph{pp}-waves of 
\cite{FF:05,FIZ:05}, see section \ref{sec:bm}, p.\ \pageref{FF-asymptotics} 
imply that $\H_{,k}$ and $\a_{k,u}$ even decay or only grow logarithmically 
for large $x$. However, again physically 
reasonable models are singular on the axis (cf.\ e.g.\ 
\eqref{pp-gyraton_axysym_4d}) or should be matched to some interior matter 
region so that the wave surface is $\R^n$ without a point or $\R^n$ with a ball removed
and hence incomplete. So again we obtain for such `gravitational' gyratons only 'completeness at infinity' but the geodesics 
could leave the exterior
region `at the inside' proceeding into the matter region. This behaviour 
again is to be considered as physically perfectly reasonable.

\section{Impulsive limit}\label{sec:il}
In this section we turn our focus to impulsive versions of the Brinkmann 
metric \eqref{bm}. Generally, impulsive gravitational waves model short but 
violent pulses of gravitational or other radiation. In 
particular, in his seminal work \cite{Pen:72}, R.\ Penrose has considered 
impulsive \emph{pp}-waves, that is spacetimes of the form 
\eqref{pp} with
\begin{align}\label{ieprofile}
 \H(x,u)=H(x)\delta(u) \,,
\end{align}
where $\delta$ denotes the Dirac function and $H$ is a function of the spatial 
variables only. Since then various methods of constructing impulsive 
gravitational waves with or without cosmological constant have been introduced, 
for an overview see e.g.\ \cite[Ch.\ 20]{GP:09}.
In particular, impulsive gravitational waves have been found to arise as 
ultrarelativistic limits of Kerr-Newman and other static spacetimes which make 
them interesting models for quantum scattering in general relativistic 
spacetimes. 

More generally, impulsive NPWs (iNPWs), i.e., \eqref{npw} with \eqref{ieprofile} 
have been considered in \cite{SS:12,SS:15}. In all these models, which are 
impulsive versions of special cases of \eqref{bm} with the off-diagonal 
terms $\a_i$ vanishing, the field equations put no restriction on the 
$u$-behaviour of the profile function $\H$, see Section~\ref{sec:bm}. Hence the 
most straightforward approach to impulsive waves in this class of solutions
indeed is to view them as impulsive limits of sandwich waves with ever 
shorter but stronger profile function which precisely leads to 
\eqref{ieprofile}.

Here we are, however, mainly interested in impulsive versions of the 
full \emph{pp}-wave spacetimes \eqref{bm}, which in particular includes 
impulsive versions of gyratonic \emph{pp}-waves \eqref{pp-gyraton}. Here the 
situation is 
more subtle as detailed in \cite{FF:05,FIZ:05}, where such geometries have been 
considered  along with their extended versions. A more detailed 
discussion of four-dimensional geometries with a flat transverse space in the 
form (\ref{ppGyrPolar}) was given recently in \cite[Sec.~7]{PSS:14}. 
Since this discussion also applies to the general case and leads to our model 
of the impulsive full \emph{pp}-wave metric we briefly recall it here. 
To begin with we introduce the convenient quantity
\begin{equation} \label{defomegaJ}
\omega(\rho,\varphi,u) \equiv \frac{J_{,\rho}(\rho,\varphi,u)}{2\rho}\,,
\end{equation}
such that the vacuum field equations take the form (with $\triangle$ denoting 
the flat Laplacian)
\begin{equation}
\omega_{,\varphi} = 0 \,, \qquad   \omega_{,\rho} = 0 \,, \qquad \triangle\, \H 
= 4\,\omega^2 +\frac{2}{\rho^2}\,J_{,u\varphi} \,, \label{fieldeqqq2}
\end{equation}
implying ${\omega=\omega(u)}$ which corresponds to a \emph{rigid rotation}. 
Relation (\ref{defomegaJ}) immediately gives
\begin{equation}\label{explicit_metric_function_J}
J = \omega(u)\, \rho^2+\chi (u,\varphi)\,,
\end{equation}
where ${\chi (u,\varphi)}$ is an arbitrary $2\pi$-periodic function in 
$\varphi$. 
Taking (\ref{explicit_metric_function_J}) and a suitable ansatz for $\H$,
\begin{equation}\label{explicit_metric_function_H}
\H = \omega^2(u)\,\rho^2+2\,\omega(u)\,\chi (u,\varphi)+ \H_0(u,\rho,\varphi)\,,
\end{equation}
the remaining field equation in (\ref{fieldeqqq2}) becomes
\begin{equation}\label{Poisson}
\triangle\, \H_0 = \rho^{-2}\,\Sigma\,, \qquad \hbox{with}\qquad 
\Sigma(u,\varphi)\equiv 2(\chi_{,u\varphi}-\omega\,\chi_{,\varphi\varphi})\,.
\end{equation}
Removing the rigid rotation by the the natural global \emph{gauge} 
${\omega=0}$, and using the splitting
\begin{equation}\label{ro-profilessosig}
\H_0(u,\rho,\varphi)= \tilde \H_0(\rho,\varphi)\,\chi_H(u)\,, \qquad \chi 
(u,\varphi) = \tilde \chi(\varphi) \,\chi_J(u)+\Phi(\varphi)\,,
\end{equation}
we obtain ${\Sigma(u,\varphi)= 2\tilde \chi_{,\varphi}(\varphi) 
\,{\chi_J}_{,u}(u)}$ and equation (\ref{Poisson}) takes the form
\begin{equation}\label{SigmaGeneral}
\triangle \tilde \H_0(\rho,\varphi)\,\chi_H(u)=\frac{2}{\rho^2}\ \tilde 
\chi_{,\varphi}(\varphi) \,{\chi_J}_{,u}(u)\,.
\end{equation}

If ${\Sigma=0}$, i.e., ${\tilde \chi(\varphi)=const.}$, equation 
(\ref{SigmaGeneral}) reduces to ${\triangle\, \H_0=0}$ and there is no 
restriction on the $u$-dependence of $\H_0$ and $J$. In particular, the energy 
profile 
$\chi_H(u)$ and the angular momentum density profile $\chi_J(u)$ can be taken 
independently of each other. In \cite{PSS:14} it was demonstrated that the 
curvature is proportional to $\chi_H$ and ${\chi_J}_{,u}$ which leads to impulsive waves by setting $\chi_H(u)$ to be 
proportional to the Dirac $\delta$ but using a box-like profile for~$\chi_J(u)$.

However, in the case when $\Sigma\not=0$ there occurs a coupling of the 
profile functions. Indeed the supports of $\chi_H(u)$ and 
${\chi_J}_{,u}(u)$ have to coincide since otherwise both sides of 
\eqref{SigmaGeneral} have to vanish individually, leading to the vanishing of 
$\tilde \chi_{,\varphi}(\varphi)$ and hence $\Sigma$. In particular, the 
box profile in the angular momentum density $\chi_J(u)$ leads to two Dirac 
deltas in the energy density. 

Moreover, we can of course combine such a coupled solution with specific 
homogeneous solutions. Hence it is most natural and physically relevant to 
prescribe a general box-like profile for the angular momentum density and a 
delta-like profile for the energy density of the form
\begin{equation}
 \H(x,u)=H(x)\delta_{\alpha,\beta}(u) \,, \quad \a_i(x,u)=a_i(x)\vartheta_L(u) 
\,,
\end{equation}
where we define (see Figure \ref{profilefig}, below)
\begin{equation}
 \delta_{\alpha,\beta}(u)=\alpha \delta(u) + \beta \delta(u-L), 
\quad\text{and}\quad
 \vartheta_L(u)=\frac{1}{L}\big(\Theta(u)-\Theta(u-L)\big) \,.
\end{equation}
Here $\alpha$, $\beta$, and $L>0$ are some constants, $\delta$ denotes the Dirac 
measure and $\Theta$ is the Heaviside function. This ansatz 
covers the coupled 
case ($\alpha=1/L,\beta=-1/L$) as well as all the models studied 
in \cite{YZF:07} ($p$-gyraton: $\alpha=0=\beta$, AS-gyraton: $a_i=0=\beta$,
$a$-gyraton: $\alpha>0$, $\beta=0$, $b$-gyraton: $\alpha=0$, $\beta>0$),
which all arise from specific combinations of homogeneous solutions.

So the impulsive full \emph{pp}-wave metric we will consider in the 
rest of our work is explicitly given by 
\begin{equation}\label{ibm}
 ds^2=h_{ij}dx^idx^j-2dudr+H(x)\delta_{\alpha,\beta}
(u)du^2+2a_i(x)\vartheta_L(u)dudx^i \,.
\end{equation}

\begin{figure}[htb]
\begin{center}
\includegraphics[scale=1.1]{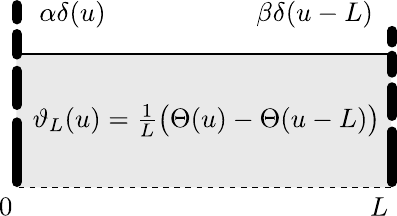}
\end{center}
\caption{Illustration of the box-like profile with accompanying delta-spikes.}
\label{profilefig}
\end{figure}

Of course, off the wave zone (given by $u\in[0,L]$) the spacetime is just 
the product of the Riemannian wave surface $(N,h)$ with flat $\R^2$. From now 
on we will assume $(N,h)$ to be complete and call $M_0=N\times\R^2$ the 
background of the impulsive wave \eqref{ibm}, which is then complete as well.
From (\ref{Psi2s})--(\ref{Psi4ij}) we immediately observe that the 
components of boost weight ${-1}$ and of ${-2}$, namely,
\begin{align}
\Psi_{3T^{i}} &= \tfrac{1}{n}\,m_{(i)}^i\,h^{kl}a_{[i,k]||l}\,\vartheta_L \,, 
\label{Psi3Tj_imp}\\
\tilde{\Psi}_{3^{ijk}} &= 
m_{(i)}^im_{(j)}^jm_{(k)}^k\,\Big(a_{[k,j]||i}-\tfrac{1}{n-1}\,h^{mn}\,\big(h_{
ij}a_{[k,m]||n}-h_{ik}a_{[j,m]||n}\big)\Big)\,\vartheta_L \,, \label{Psi3ijk_imp}\\
\Psi_{4^{ij}} &= 
m_{(i)}^im_{(j)}^j\,\Big(-\tfrac{1}{2}H_{||ij}\,\delta_{\alpha,\beta}+a_{(i||j)}\,\delta_{L^{-1},-L^{-1}}+h^{mn}a_{[m,i]}
a_{[n,j]}\,{\vartheta_L}^2 \nonumber \\
& \hspace{25.0mm} -\tfrac{1}{n}\,h_{ij} 
h^{kl}\big(-\tfrac{1}{2}H_{||kl}\,\delta_{\alpha,\beta}+a_{k||l}\,\delta_{L^{-1},-L^{-1}}+h^{mn}a_{[m,k]}a_{[n,l]}\,{
\vartheta_L}^2
\big)\Big) \,,\label{Psi4ij_imp}
\end{align}
are only non-trivial in the wave zone while, in general, the rest of 
the spacetime corresponds to the type D background.

\section{Completeness of the impulsive limit}\label{sec:cil}
We first review previous results on the completeness of impulsive gravitational 
waves. In the simplest case of classical impulsive \emph{pp}-waves, 
i.e., \eqref{pp} with \eqref{ieprofile},  the spacetime 
is flat Minkowski space off the single wave surface $\{u=0\}$ where the 
curvature is concentrated. Consequently, the geodesics for impulsive 
\emph{pp}-waves have been derived in the physics literature (see e.g.\ 
\cite{FPV:88}) by matching the geodesics of the background 
on either side of the wave in a heuristic manner---the geodesic equation 
contains nonlinear terms, ill-defined in distribution theory. This approach, in 
particular, leaves it open whether the geodesics cross the wave 
surface at all.   

In \cite{KS:99a, KS:99b} this question has been answered in the affirmative 
using a regularisation approach within the theory of nonlinear distributional 
geometry (\cite[Ch.\ 4]{GKOS:01}) based on algebras of generalized functions 
(\cite{C:85}). In this way a completeness result for all impulsive 
\emph{pp}-waves, i.e., for \emph{all} smooth profile functions $H$, was 
achieved although this aspect was not emphasised in the original works. Observe 
that this contrasts the completeness results in the extended case 
(Corollary \ref{NPWc}) where the spatial asymptotics of $\H$ enters decisively.
Moreover this approach in a limiting process 
establishes that the geodesics in the entire spacetime are indeed the straight 
line geodesics of the background which are refracted by the impulse to become 
broken and possibly discontinuous.

More generally in \cite{SS:12, SS:15} geodesics in iNPWs, i.e., \eqref{npw} 
with \eqref{ieprofile},  were investigated. Again using a regularisation 
approach it was proven that if the wave surface $N$ is \emph{complete}, then the 
iNPW is geodesically complete irrespective of the behaviour of the profile 
function $H$. Again this is in contrast to the extended case where the 
completeness depends crucially on the spatial asymptotic behaviour of the 
profile function $\H$, see Section ~\ref{sec3}. Moreover, the geodesics in the 
limit 
again are geodesics of the background, which are refracted by the impulse, see 
also \cite{FIZ:05} for a heuristic argument. 
More precisely, using a fixed point 
argument it was shown in 
\cite{SS:12} that in the regularised iNPW
\begin{align}\label{inpw}
 ds_\eps^2=h_{ij}dx^idx^j-2dudr+\delta_\eps(u) H(x)du^2 \,,
\end{align}
with $\delta_\eps$ a standard mollifier\footnote{For a precise definition see 
\eqref{molli}, below.}, 
geodesics are complete in the following sense: For each point $p\in M$ `in 
front' of the impulsive wave, i.e., $u<0$ and
each tangent direction $v\in T_pM$ we consider the geodesic $\gamma_\eps$ of 
\eqref{inpw} starting in $p$ into direction $v$. If $\gamma_\eps$ reaches 
the regularised wave zone given by $\{|u|\leq 
\eps\}$ then there is an $\epsilon_0>0$ such that for 
all $0<\epsilon\leq\epsilon_0$ the geodesic $\gamma_\eps$ passes through the 
regularised wave zone and continues as (complete) geodesic of the background 
`behind' the 
impulsive wave. This result has been rephrased in the language of nonlinear 
distributional geometry in \cite{SS:15}, which allows to omit the reference to 
the initial data in the final completeness statement.

\medskip

As is well known, classical impulsive \emph{pp}-waves and more generally 
non-expanding as well as expanding impulsive gravitational waves propagating in 
constant curvature backgrounds have also been described by a continuous form of 
the metric, see e.g. \cite[Ch.\ 20]{GP:09}. Actually these metrics are locally 
Lipschitz continuous and hence the geodesics equations possess locally 
bounded but possibly discontinuous right hand sides. Employing the solution 
concepts of Carath\'eodory and Filippov (\cite{F:88}), respectively, these 
systems of ODEs have been recently investigated leading to the following 
results: The geodesics are complete and of $C^1$-regularity in classical 
\emph{pp}-waves (\cite{LSS:14}), non-expanding (\cite{PSSS:15}) and expanding 
(\cite{PSSS:16}) impulsive waves propagating on Minkowski, de Sitter, and 
anti-de Sitter backgrounds. However, so far no continuous form of the impulsive 
full \emph{pp}-wave metric or merely of the gyratonic \emph{pp}-wave metric 
has been found. 

Geodesic completeness for non-expanding impulsive gravitational waves 
in (anti-)de Sitter space has also been proven in the distributional picture in 
\cite{SSLP:15} using a regularisation approach and a fixed point argument in 
a spirit similar to the present article. Finally a proof of geodesic 
completeness of impulsive gyratonic \emph{pp}-waves has been sketched in 
\cite[Section\ VIII]{PSS:14}.   
\medskip

In the following we provide our main result which establishes completeness 
of the impulsive full \emph{pp}-wave metric \eqref{ibm} using a 
regularisation approach.

To begin with we consider the regularised metric
\begin{equation}\label{ribm}
 ds^2_\varepsilon=h_{ij}dx^idx^j-2dudr+H(x)\delta^\varepsilon_{\alpha,\beta}
(u)du^2+2a_i(x)\vartheta^\varepsilon_L(u)dudx^i\,,
\end{equation}
where we have regularised the profile functions replacing $\delta$ by a
standard mollifier
\begin{align}\label{molli}
 \delta_\varepsilon(x)=\frac{1}{\varepsilon}\,\phi\left(\frac{x}{
\varepsilon}
\right) 
\end{align}
with $\phi$ a smooth function supported in $[-1,1]$ with unit integral. 
Moreover
we have regularised the Heaviside function by the primitive of
$\delta_\varepsilon$, i.e., 
replacing $\Theta$ by 
\begin{align}
 \Theta_\epsilon(x)\equiv\int_{-1}^x\delta_\epsilon(t)\,\d t. 
\end{align}
More explicitly we set
\begin{equation}
 \delta^\varepsilon_{\alpha,\beta}
(u)=\alpha\delta_\varepsilon(u)+\beta\delta_\varepsilon(u-L) \,,
 \quad\text{and}\quad 
\vartheta^\varepsilon_L(u)=\frac{1}{L}
\big(\Theta_\varepsilon(u)-\Theta_\varepsilon(u-L)\big)\,.
\end{equation}

In the following we will prove that any geodesic
in the regularised impulsive full \emph{pp}-wave metric \eqref{ribm}
that reaches the wave zone given by ${\{u\in[-\eps,L+\eps]\}}$ will pass 
through it, provided that $\eps$ is small enough. This will lead to our main 
result on the completeness of the impulsive full \emph{pp}-wave metric 
\eqref{ibm} which we state at the end of this section. In the final section 
\ref{sec:lim} we will relate these complete geodesics to the geodesics of the 
background.

To begin with we give the explicit form of the geodesic equations for the 
metric \eqref{ribm}. Observe that also in the present case the $u$-equation is 
trivial and hence we may use a rescaling as in \eqref{ng} to write any geodesic
not parallel to, or contained in the impulsive wave surface 
$\{u=0\}$\footnote{We will deal with these (simple) geodesics separately.} 
as 
\begin{equation}\label{gammaeps}
 \gamma_\eps(s)=(x^i_\eps(s),s,r_\eps(s)).
\end{equation}
 Now we 
obtain, cf.\ \eqref{ge}
\begin{align}\label{eq:rgeoe}
\ddot r_\varepsilon&=\vartheta^\varepsilon_L\,a_{(i||j)}\,\dot 
x^i_\varepsilon\dot x^j_\varepsilon
   -\Big(g_\varepsilon^{ri}\vartheta^\varepsilon_L\,\big(a_{i,j}
-a_{j,i}\big)-H_{,j}\,\delta^\varepsilon_{\alpha,\beta}\Big)\,\dot 
x^j_\varepsilon \nonumber\\ 
  & \qquad\quad
 -\Big(g_\varepsilon^{ri}\big(a_i\,\delta^\varepsilon_{L^{-1},-L^{-1}} 
-\tfrac{1}{2}
 H_{,i}\,\delta^\varepsilon_{\alpha,\beta}\big)-\tfrac{1}{2}H\,
 (\delta^\varepsilon_{\alpha,\beta})_{,u}\Big) \,, \\ \label{eq:regeoe}
\ddot x_\varepsilon^i&=-{\Gamma^{(N)}}^i_{jk}\,\dot x_\varepsilon^j\dot
x_\varepsilon^k-\vartheta^\varepsilon_L\,h^{ik}\big(a_{k,j}-a_{j,k}\big)\,\dot 
x^j
           -\tfrac{1}{2}h^{ik}\big(2a_k\,\delta^\varepsilon_{L^{-1},-L^{-1}}
-H_{,k}\,\delta^\varepsilon_{\alpha,\beta}\big) \,, 
\end{align}
where $g_\varepsilon^{ri}=h^{ik}a_k\,\vartheta_L^\varepsilon$.
\medskip 

As in the extended case the $r$-equation can simply be integrated once 
the $x$-equations are solved and completeness of the geodesics is determined by
completeness of the solutions to the spatial equations, cf.\ Proposition 
\ref{prop:fog}. The latter again take the form of the equations of motion on 
the Riemannian manifold $(N,h)$, now with an external force term depending on 
time, 
velocity and the regularisation parameter $\eps$. More explicitly we may rewrite
equation \eqref{eq:regeoe} in the form, cf.\ \eqref{ODE}
\begin{align}\label{ODEe}
D_{\dot x_\eps}\dot x_\eps(s)=
 F^\eps_{(x_\eps(s),s)}\dot x_\eps(s)
 +X^\eps_{(x_\eps(s) , s) } \ ,
\end{align}
where $D$ denotes the connection on $(N,h)$ and we have set
$F^{i\,\eps}_j=-\vartheta^\varepsilon_L\,h^{ik}(a_{k,j}-a_{j,k
})$ and 
$X^{i\,\eps}=-\tfrac{1}{2}h^{ik}\big(2a_k\,\delta^\varepsilon_{L^{-1},-L^{-1}}
-H_{,k}\,\delta^\varepsilon_{\alpha,\beta}\big)$.

Now for fixed $\eps$ the solutions will be complete by Corollary \ref{cor:cbm}
\emph{provided} $H$ and $a_i$ show a suitable asymptotic behaviour. 
But here we aim at a result for general $H$ and $a_i$ and so we have to 
take a different approach. Indeed, for fixed $\eps$ by ODE-theory we have a 
local solution $x_\eps$ for any initial condition taken at the say `left' 
boundary of the wave zone $u=-\eps$, see Figure \ref{fig:one-sol}. However, the time of 
existence of such a solution will in general depend upon the regularisation 
parameter $\eps$ and could shrink to zero if $\eps\to 0$. We will prove that 
this is \emph{not} the case. More precisely, applying a fixed point argument we 
will show that such solutions $x_\eps$ have a uniform (in $\eps$) lower bound 
$\eta$ on their time of existence, which will at least for small $\eps$ allow 
them to cross the regularisation region of the first $\delta$-spike, i.e., 
$|u|\leq\eps$. Once they reach $u=\eps$ they are subject to equations 
\eqref{eq:regeoe} with only the $\vartheta^\eps$-terms being non-trivial. In 
other words we have to deal with \eqref{ODEe} with $X^\eps$ vanishing. In this 
situation we may now apply the completeness results established in Section 
\ref{sec3}.
More precisely, we appeal to Theorem \ref{crs-thm}(1) whose assumptions hold 
anyway in our case and we obtain that the solution $x_\eps$ will continue at 
least until it reaches the regularisation region of 
the second $\delta$-spike at $u=L-\eps$. There we can, however, reapply our 
fixed point argument to secure that $x_\eps$ reaches $u=L+\eps$ and hence 
leaves the entire wave zone to enter the background region `behind' the wave.
\medskip

We will now state and prove the fixed point argument. To simplify notations we 
will, instead of dealing with equation \eqref{eq:regeoe} directly, consider the 
following  model initial value problem
\begin{align}\label{ivp1}
 \ddot x_\varepsilon&=F_1(x_\varepsilon,\dot 
x_\varepsilon)+F_2(x_\varepsilon)\delta_\varepsilon + 
 F_3(x_\varepsilon,\dot x_\varepsilon)\frac{1}{L}\Theta_\varepsilon \,, \\ 
 \label{ivp2}
 x_\varepsilon(-\varepsilon)&=x_0^\eps,\quad \dot 
x_\varepsilon(-\varepsilon)=\dot{x}_0^\eps \,.
\end{align}
Also we will write $(x_\eps)_{\eps\in(0,1]}$ or briefly 
$(x_\eps)_{\eps}$ to denote nets (sequences). Now we have
\begin{prop}\label{ex}
 Let $F_1, F_3\in C^\infty(\R^{2n},\R^n)$, $F_2\in C^\infty(\R^{n},\R^n)$, let 
$x_0,\dot x_0\in\R^n$, let $(x_0^\eps)_\eps$, 
$(\dot{x}_0^\eps)_\eps$ in $\R^n$ such that $x_0^\eps\to x_0$ and 
$\dot{x}_0^\eps\to\dot{x}_0$ for $\eps\searrow 0$  and
 let $b,c>0$. Define $I_1=\{x\in\R^n: |x-x_0|\leq b\}$, $I_2=\{x\in\R^n: |\dot 
x-\dot x_0|\leq c+K\|F_2\|_{\infty,I_1}\}$ and
 $I_3=I_1\times I_2$, where $K$ is a bound on the $L^1$-norm of 
$(\delta_\varepsilon)_\varepsilon$.  
 Furthermore set 
 \begin{equation}
   \eta=\min\left(1,\frac{b}{C_1},\frac{c}{C_2},\frac{L}{2}\right) ,
 \end{equation}
 where $C_1=2+|\dot x_0|+\|F_1\|_{\infty,I_3}+K\|F_2\|_{\infty,I_1}
+\frac{K}{L}\|F_3\|_{ \infty,I_3}$ and $C_2=1+\|F_1\|_{\infty,I_3}+
\frac{K}{L}\|F_3\|_{\infty,I_3}$. Finally, let $\eps_0'$ be such that 
$|x_0^\eps-x_0|\leq \eta$ and 
$|\dot{x}_0^\eps-\dot{x}_0|\leq \eta$ for all $0<\eps\leq\eps_0'$. Then the 
initial value problem \eqref{ivp1},
\eqref{ivp2} has a unique solution $x_\varepsilon$ on 
$I_\varepsilon=[-\varepsilon,\eta-\varepsilon]$ with
$(x_\varepsilon(I_\varepsilon),\dot x_\varepsilon(I_\varepsilon))
 \subseteq I_3$.
\end{prop}

\begin{proof}
We aim at applying Weissinger's fixed point theorem  (\cite{Wei:52}) to the
solution operator
\begin{align}\nonumber
 A_\varepsilon(x)(t):=x_0^\eps&+\dot{x}_0^\eps(t+\varepsilon)\\&+
 \int_{-\varepsilon}^t\int_{-\varepsilon}^s
 \left( F_1(x(\sigma),\dot x(\sigma)) 
 + 
 F_2(x(\sigma))\delta_\varepsilon(\sigma)
 +
 F_3(x(\sigma),\dot x(\sigma)) \Theta_\varepsilon(\sigma)\right)\d\sigma\, \d s
\end{align}
on the complete metric space
\begin{equation}
 X_\varepsilon:=\{x\in C^1([-\varepsilon,\eta-\varepsilon]):\ 
 (x,\dot x)([-\varepsilon,\eta-\varepsilon])\subseteq I_3\} \,,
\end{equation}
where we use the norm $\|x\|_{C^1}=|x|+|\dot x|$.

To begin with we show that $A_\varepsilon$ maps $X_\varepsilon$ to itself.
Indeed we have for $x\in X_\varepsilon$
\begin{align}\nonumber
 |A_\varepsilon(x)(t)-x_0|
 &\leq |x_0^\eps-x_0| +  \eta(|\dot{x}_0^\eps-\dot{x}_0| + |\dot{x}_0|)+
  \eta^2\|F_1\|_{\infty,I_3}\\
  &\quad 
+\eta\|F_2\|_{\infty,I_1}\|\delta_\varepsilon\|_{L^1} + \eta^2\frac{1}{L}
\|F_3\|_{\infty,I_3}
\|\Theta_\varepsilon\|_\infty\\\nonumber
 &\leq\eta C_1\ \leq b\,,\ \text{and}\\\nonumber
  |\frac{d}{dt} A_\varepsilon(x)(t)-\dot x_0|
 &\leq |\dot{x}_0^\eps-\dot{x}_0| + \eta 
(\|F_1\|_{\infty,I_3}+\frac{K}{L}\|F_3\|_{\infty,I_3})
  +K\|F_2\|_{\infty,I_1}\\ &\leq\ c + K\|F_2\|_{\infty,I_1} \,.
\end{align}
Moreover we have for $x,y\in X_\varepsilon$
\begin{align}
 |A^n_\varepsilon(x)(t) - A^n_\varepsilon(y)(t)|
 \leq& \mathop{Lip}(F_1,I_3)\|x-y\|_{C^1}\frac{\eta^{2n}}{(2n)!}
  +\mathop{Lip}(F_2,I_1)K |x-y|\frac{\eta^{2n-1}}{(2n-1)!}\nonumber\\
 &+\mathop{Lip}(F_3,I_3)\frac{K}{L}\|x-y\|_{C^1}\frac{\eta^{2n}}{(2n)!}
\nonumber\\
 |\frac{d}{dt}A^n_\varepsilon(x)(t) - \frac{d}{dt}A^n_\varepsilon(y)(t)|
 \leq& \mathop{Lip}(F_1,I_3)\|x-y\|_{C^1}\frac{\eta^{2n-1}}{(2n-1)!}
  +\mathop{Lip}(F_2,I_1)K |x-y|\frac{\eta^{2n-2}}{(2n-2)!}\nonumber\\
 &+\mathop{Lip}(F_3,I_3)\frac{K}{L} \|x-y\|_{C^1}\frac{\eta^{2n-1}}{(2n-1)!} \,,
\end{align}
where $\mathop{Lip}(F_i,I_j)$ denotes a Lipschitz constant for $F_i$ on
$I_j$. So we have 
\begin{equation}
 \|A^n_\varepsilon(x)-A^n_\varepsilon(y)\|_{C^1}\leq C \frac{\eta^{2n}}{(2n)!}
\|x-y\|_{C^1}
\end{equation}
and since $\sum\frac{\eta^{2n}}{(2n)!}$ converges, we obtain a unique fixed 
point
of $A_\varepsilon$ on $X_\varepsilon$ hence a net of unique solutions 
$x_\varepsilon$ of \eqref{ivp1}, \eqref{ivp2} defined on
$[-\varepsilon,\eta-\varepsilon]$ which together with its derivatives is
uniformly bounded in $\varepsilon$.
\end{proof}

We will now detail the procedure envisaged prior to Proposition \ref{ex} to 
obtain our main result. Fix a point 
$$p =(x_p,u_p,r_p)\quad\text{in $M$}$$ 
lying say `before' the wave zone\footnote{The entire
argument is precisely the same in the `time-reflected' case when $p$ is assumed 
to lie `behind'  the wave zone, i.e., in $\{u>L\}$.}, i.e., $u_p<0$, and a 
vector $v$ 
in $T_pM$. In the following we will most of the time simplify notations by 
omitting the index from the $x$-component as we have done in the model initial 
value problem \eqref{ivp1}, \eqref{ivp2} and write e.g.\ $x_p$ instead of 
$x^i_p$. Now, without loss of generality we may 
assume $\eps$ to be so small that $p$ also lies `before' the 
\emph{regularised} wave zone, i.e., in $\{u<-\eps\}$ 
and hence in a the region of $(M,g_\eps)$ which coincides with the background spacetime 
$M_0=N\times\R^2$ `before' the impulse. Now we consider the 
geodesic $\gamma(s)=(x(s),u(s),r(s))$ starting at $p$ in direction $v$ \emph{in 
the background} spacetime $M_0$ which we will from now on call our `seed 
geodesic'. By virtue of the geodesic equations in the background $M_0$
\begin{equation}\label{eq:bgrgeo}
  D_{\dot{x}}\dot{x} = 0,\quad \ddot u=0,\quad  \ddot{r}=0, 
\end{equation}
we see again that $u(s)$ is affine and hence 
it suffices to consider the case of strictly increasing $u(s)$. Indeed 
otherwise the `seed geodesic' will never reach the (regularised) wave zone 
being either confined to the null surface $P(u_p)$ (cf.\ Section \ref{sec:bm}),
or even moving away from the wave zone and hence in any case be (forward) 
complete. So we may without loss of generality write the seed geodesic in the 
form \eqref{ng}, i.e., $\gamma(s)=(x(s),s,r(s))$ or briefly as 
$\gamma(s)=(x(s),r(s))$.

Now $\gamma$ will reach the wave zone of the impulsive wave, i.e., $s=0$ in 
finite time and it is 
convenient to introduce the data of $\gamma$ at this instance as
\begin{equation}\label{eq:0-data}
 \gamma(0)=(x_0,0,r_0),\ \dot\gamma(0)=(\dot x_0,1,\dot r_0)\,.
\end{equation}
Now we start to think of the `seed geodesic' $\gamma$ also of being a geodesic 
in the regularised space time \eqref{ribm}. In fact it will reach the 
regularised wave zone at $u=-\eps$ with data
\begin{equation}\label{eq:e-data}
 \gamma(-\eps)=(x^\eps_0,\eps,r^\eps_0),\ 
 \dot\gamma(\eps)=(\dot x^\eps_0,1,\dot r^\eps_0)\,.
\end{equation}
Using this data we solve the initial value problem for the geodesics in the 
regularised spacetime \eqref{ribm}, that is we consider the system 
\eqref{eq:rgeoe}, \eqref{eq:regeoe} with data \eqref{eq:e-data}. Now by 
smoothness of the `seed geodesic' $\gamma$ the data \eqref{eq:e-data} 
converges to the data \eqref{eq:0-data}, in particular,
\begin{equation}
 x^\eps_0\to x_0\quad\text{and}\quad\dot x_0^\eps\to\dot x_0\,,
\end{equation}
and we may apply Proposition \ref{ex} to obtain a solution $x_\eps$ of \eqref{eq:regeoe}, 
\eqref{eq:e-data} on $(-\infty,\eps)$, provided $\eps\leq \eta/2$. Hence we 
obtain also a solution $r_\eps$ of \eqref{eq:rgeoe} with data 
\eqref{eq:e-data} hence a geodesic $\gamma_\eps$ which coincides with the `seed 
geodesic' $\gamma$ up to $s=-\eps$ and exists until it leaves the 
regularised first $\delta$-spike at $s=\eps$ and we denote the 
corresponding data by
\begin{equation}\label{data1}
 \gamma_\eps(\eps)=(x_1^\eps,\eps,r_1^\eps),\ 
 \dot\gamma_\eps(\eps)=(\dot x_1^\eps,1,\dot r_1^\eps)\,.
\end{equation}
As discussed earlier on $[\varepsilon,L-\varepsilon]$ the geodesic equation 
\eqref{eq:regeoe} reduces to
\begin{equation}\label{eq-between}
  \ddot x_\varepsilon^i=-{\Gamma^{(N)}}^i_{jk}\dot x_\varepsilon^j\dot
x_\varepsilon^k-\frac{1}{L}h^{ik}\big(a_{k,j}-a_{j,k}\big)\,\dot
x_\varepsilon^j \,,
\end{equation}
whose right hand side is actually independent of $\varepsilon$. 
However, we have to solve \eqref{eq-between} with the $\eps$-dependent data 
\eqref{data1}. Anyway by Theorem \ref{crs-thm}(1) we obtain a solution 
$x_\eps$ which extends our prior solution from $s=\eps$ up to 
$s=L-\eps$ and since by Proposition \ref{ex}
the data $x^\eps_1$ and $\dot x^\eps_1$ are uniformly bounded (in $\eps$) 
the solutions will be uniformly bounded as well, in particular this applies to the data at $s=L-\eps$,
\begin{equation}\label{eq:L-e-data}
x^\eps(L-\eps)=x_2^\eps,\ \dot x^\eps(L-\eps)=\dot x_2^\eps\,.
\end{equation}
We now wish to reapply Proposition \ref{ex} on the interval 
$[L-\eps,L-\eps+\eta]$ 
and so we need the data \eqref{eq:L-e-data} to even converge.
This however, follows from continuous dependence of solutions to ODEs once we have established that 
the data \eqref{data1} converges which we do next.

\begin{lem}\label{lem-lim-xeps}
 Let $(x_{\eps})_\eps$ given by Proposition \ref{ex}. Then 
 \begin{enumerate}
  \item[(i)] $\displaystyle \sup_{t\in[-\eps,\eps]}|x_\eps(t)-x_0| = O(\eps)$,
  \item[(ii)] $\displaystyle\dot{x}_\eps(\eps)\to \dot{x}_0 - F_2(x_0)$ as 
$\eps\to0$.
 \end{enumerate}
\end{lem}

\begin{proof}
 (i) On $[-\eps,\eps]$ the solution $x_\eps$ is given by Proposition \ref{ex} 
and  can be expressed by
\begin{align}
 x_\eps(t) &= x^\eps_0 + \dot{x}^\eps_0(t+\eps) \nonumber\\&+ 
\int_{-\eps}^t\int_{-\eps}^s 
F_1(x_\eps(\sigma),\dot{x}_\eps(\sigma)) + 
F_2(x_\eps(\sigma))\delta_\eps(\sigma) + 
F_3(x_\eps(\sigma),\dot{x}_\eps(\sigma))\theta_\eps(\sigma)\d \sigma \d 
s\,.\label{eq:xeps}
\end{align}
Consequently, 
\begin{align}\nonumber
 |x_\eps(t)-x_0| \leq& |x^\eps_0 - x_0| + 2\eps (|\dot{x}^\eps_0-\dot{x}_0| + 
|\dot{x}_0|)\\& + 4\eps^2\|F_1\|_{\infty, I_3} + 
2\eps \|F_2\|_{\infty,I_1} K + 4\eps^2 K \|F_3\|_{\infty,I_3}\,,
\end{align}
where we used the uniform boundedness of $(x_\eps)_\eps$ and 
$(\dot{x}_\eps)_\eps$ established in Proposition \ref{ex}. 

To obtain (ii) we differentiate 
\eqref{eq:xeps}, insert $t=\eps$, and then we estimate
\begin{align*}
 |\dot{x}_\eps(\eps)-\dot{x}_0 - F_2(x_0)| &\leq |\dot{x}^\eps_0-\dot{x}_0| + 
2\eps \|F_1\|_{\infty,I_3}\\ &\quad + 
|\int_{-\eps}^\eps (F_2(x_\eps(s)) - F_2(x_0))\delta_\eps(s)\d s| + 2 K \eps 
\|F_3\|_{\infty,I_3}\\
&\leq |\dot{x}^\eps_0-\dot{x}_0| + 2\eps \|F_1\|_{\infty,I_3}\\
&\quad + K 
\sup_{s\in[-\eps,\eps]} |F_2(x_\eps(s)) - F_2(x_0)| + 2 K 
\eps \|F_3\|_{\infty,I_3}\,,
\end{align*}
where we have used that $\int_{-\eps}^\eps \delta_\eps(s)\d s = 1$ in the first 
inequality and (i) 
to see that the first term in the final line converges to zero as 
$\eps\searrow 0$.
\end{proof}

We will explicitly give the limit in (ii) for our case in Section \ref{sec:lim}, below. 
For the time being we are in the position to reapply Proposition \ref{ex} to 
obtain a solution $\gamma_\eps=(x_\eps, r_\eps)$ to
\eqref{eq:regeoe}, \eqref{eq:rgeoe} with data \eqref{eq:L-e-data} on the domain 
$[L-\eps,L+\eps]$, again provided that
$\eps\leq\eta/2$. Moreover $x_\eps$ is uniformly bounded in $\eps$ 
together with its derivative,  which, in particular, applies to the data at 
$s=L+\eps$,
\begin{equation}\label{eq:L+e-data}
 x_\eps(L+\eps)=x_3^\eps,\ \dot {x}_\eps(L+\eps)=\dot x_3^\eps\,.
\end{equation}
But now we have reached the background spacetime `behind' the regularised wave zone
and the solutions just obtained can be continued as solutions $x_\eps$ of 
the background geodesic equations \eqref{eq:bgrgeo} with data 
\eqref{eq:L+e-data}. By completeness of the background $M_0$ these solutions 
extend to all positive values of their parameter. Now inserting this solution 
into the geodesic equation's $r$-component \eqref{eq:rgeoe} we obtain also a 
forward complete solution $r_\eps$. Hence together we have obtained a complete 
smooth geodesic $\gamma_\eps$, which coincides with the `seed geodesic' $\gamma$ 
on $(-\infty,-\eps)$ and with a background geodesic for $u\geq L+\eps$.
Note however, that in the background 
`behind' the regularised wave zone, $\gamma_\eps$ does not coincide with a 
\emph{single} geodesic of the background since the data \eqref{eq:L+e-data},
which we feed into the background geodesic equation \eqref{eq:bgrgeo} at 
$s=L+\eps$ depends on $\eps$. Therefore the global geodesic 
$\gamma_\eps$ 
for $s\geq L+\eps$ coincides with a background geodesic 
starting at $L+\eps$ with data $x^\eps_3$, $\dot x^\eps_3$, and
$r_\eps(L+\eps)=r^\eps_3$, $\dot r_\eps(L+\eps)=\dot r_3^\eps$.

Finally, it remains to deal with the geodesics which start at points $p$ 
with $u_p\in[0,L]$. To begin with, if $u_p=0$ we start within the 
first impulsive surface in some specified direction $v\in T_pM$. If $v$ is 
tangential to the null hypersurface $P(0)$ then the corresponding geodesic 
will be either null or spacelike 
but in any case stay entirely within $P(0)$ and thus have a trivial 
$u$-component, cf.\ Section~\ref{sec:bm}. But then an inspection of the 
geodesic equation \eqref{ge} reveals that the $x$-equation coincides with the 
geodesic equation on the complete Riemannian manifold $(N,h)$ and hence its 
solution is complete. Feeding this solution into the $r$-equation (which again 
simplifies drastically) we obtain completeness. 
In case $v$ is 
transversal to $P(0)$ there is a `seed geodesic' with data \eqref{eq:0-data}
coinciding with $p$ and $v$ and we have already covered this case. Precisely 
the same argument applies in the 'time reflected' case to all points $p$ with 
$u_p=L$, i.e., which lie on the impulsive surface of the second spike. 

Finally for all points $p$ with $u_p\in(0,L)$ we may assume that $\eps$ is so 
small that $u_p\in(\eps,L-\eps)$ hence that $p$ lies in the `intermediate' 
region where the geodesic equations \eqref{eq-between} are independent of 
$\eps$. In case $v\in T_pM$ is tangential to $P(u_p)$ the geodesic again stays 
entirely in the hypersurface $P(u_p)$ and is complete by Theorem 
\ref{crs-thm}(1). In case $v$ is transversal to $P(u_p)$ again by `time 
symmetry' we have to only discuss the case of an increasing $u$-component. So 
once more by Theorem \ref{crs-thm}(1) the geodesic will reach $s=L-\eps$ and we 
may apply Proposition~\ref{ex} since the data at this instant will converge to 
the data of the corresponding solution of \eqref{eq-between} at $s=L$.       

Summing up we have proved our main result.

\begin{figure}[htb] \caption{An illustration of the construction leading 
to Theorem \ref{thm-global-ex}. The $x$-component of the seed geodesic 
$\gamma$ is drawn in bold black and the regularised geodesic given by 
Proposition \ref{ex} are drawn dotted in blue. The seed geodesic provides the 
initial data for the geodesic $x_\eps$ at $-\eps$ which on $[-\eps,\eps]$ 
solves \eqref{eq:regeoe} and continues as a solution of \eqref{eq-between} 
until $L-\eps$, where it provides new initial data for a solution 
$x_\eps$ of \eqref{eq:regeoe} until $L+\eps$ and then continues as a 
solution of the 
background geodesic equation \eqref{eq:bgrgeo} with data \eqref{eq:L+e-data}.} 
\begin{center}\label{fig:one-sol} 
\includegraphics[scale=1.1]{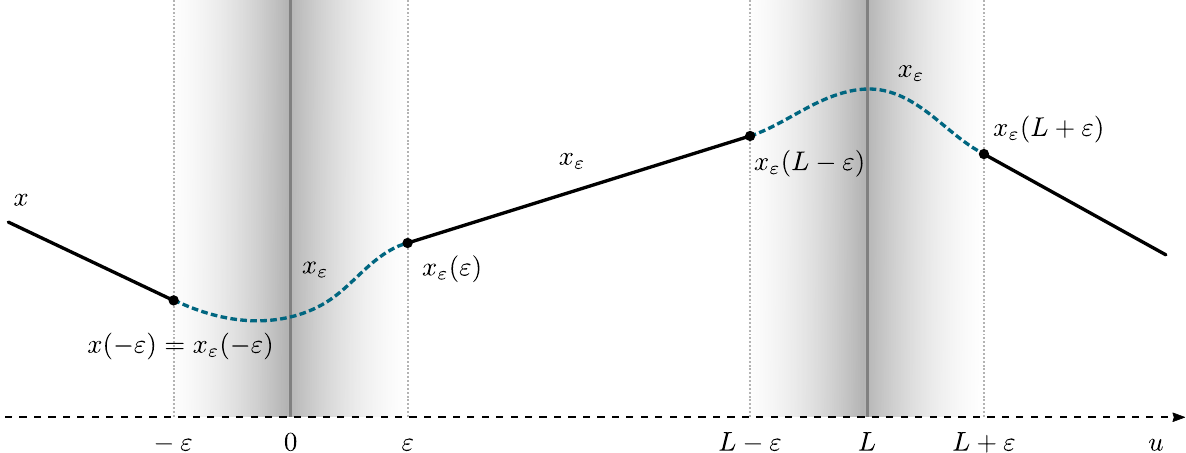} \end{center} \end{figure}

\begin{thm}\label{thm-global-ex}
 Given a point $p$ in the regularised impulsive full \emph{pp}-wave spacetime 
\eqref{ribm}
 and $v\in T_pM$. Then there exists $\eps_0$ such that the maximal unique 
 geodesic $\gamma_\eps$ starting in $p$ in direction $v$ 
 is complete, provided $\varepsilon\leq\varepsilon_0$.
\end{thm}

We now briefly discuss the case of profile functions $H$ and $a_i$ 
in the metric \eqref{ribm} possessing poles, making it necessary to remove 
them from the 
spacetime or likewise the case that the exterior solution \eqref{ribm} is 
matched to some interior non-vacuum solution for small $x$ and at least some 
$u$ interval. Recall from Section \ref{sec:bm} that such situations occur in 
physically interesting models and that this leads to an incomplete wave surface 
$N$. In such a case our method still applies but with some restrictions.
Indeed, if a `seed 
geodesic' $\gamma$ hits the wave zone at ${u=0}$ sufficiently far away from the 
poles or the matching surface to an interior solution we may first apply 
Proposition \ref{ex} with the constants $b$ and $c$ chosen so small that the 
problematic region is omitted. Then in the `intermediate region' 
$u\in[\eps,L-\eps]$ we can estimate the solution $x_\eps$ in terms of 
the data of the `seed geodesic' and the right hand side of \eqref{eq-between} 
hence independently of $\eps$, which again makes it possible to avoid the 
problematic 
region. This finally applies as well to the second application of Proposition 
\ref{ex} in the interval $[L-\eps,L+\eps]$. On the  other hand, for `seed 
geodesics' aiming too closely at poles or the matching boundary to an interior 
solution completeness cannot be guaranteed, which, however, is in complete 
agreement with physical expectations.  
\medskip

Finally, to end this section we prove an additional boundedness result for the 
global geodesics $\gamma_\eps$ of Theorem \ref{thm-global-ex}. Indeed 
local uniform boundedness of the $x$-component $x_\eps$ (and of its derivative) 
follows directly from Proposition \ref{ex} and the fact that in the 
`intermediate region' $u\in[\eps,L-\eps]$ the geodesic equation is actually 
$\eps$-independent. However, we also obtain local uniform boundedness of the 
$r$-component and hence of $\gamma_\eps$ itself, as follows from the next 
statement. 

\begin{lem}[Uniform boundedness of $r_\varepsilon$]\label{lem-v-bound}
The $r$-component $r_\eps$ of any complete geodesic $\gamma_\eps$ of 
Theorem \ref{thm-global-ex} is locally uniformly bounded in $\eps$.
\end{lem}
\begin{proof}
We first consider the first spike and hence let $s\in[-\eps,\eps]$. 
Then 
 \begin{align}
  |r_\eps(s)| \leq |r^\eps_0 - r_0| + |r_0| + \eta(|\dot{r}_0^\eps - \dot{r}_0| 
+ |\dot{r}_0|) + \int_{-\eps}^\eps 
\int_{-\eps}^\eps |\ddot{r}_\eps(\sigma)|\, \d \sigma\,\d s\,,
 \end{align}
 where for the last term we use equation \eqref{eq:rgeoe} and estimate
 \begin{align}\nonumber
   \int_{-\eps}^\eps \int_{-\eps}^\eps |\ddot{r}_\eps(\sigma)|\, \d \sigma\,\d s 
\leq& \frac{K}{L} \|Da\| 4 \eps^2 (C')^2 + \|h^{-1}\|\|a\| \frac{K^2}{L^2}\|Da\|4\eps C' 
+\|DH\|(\alpha+\beta)K 2\eps C'\\\nonumber
& +\|h^{-1}\|\|a\|^2 \frac{K^2}{L^2}4\eps + \|DH\|(\alpha+\beta)K\eps
+ 2\|DH\|(\alpha+\beta)\|\rho'\|\,,
 \end{align}
where $C':=|\dot{x}_0| + c + K\|F_2\|_\infty$. Since $\eps\leq1$ this is bounded 
independently of $\eps$. Here the 
norms are  $L^\infty$-norms over the compact set
$I_1\subseteq\R^n$ given by Proposition \ref{ex}. 

In the interval $[L-\eps,L+\eps]$ we may argue precisely in the same way. 
Finally for the intervals $(-\infty,-\eps]$, $[\eps,L-\eps]$ and 
$[L+\eps,\infty)$ one uses continuous dependence of solutions to ODEs 
on the initial conditions and the convergence of the data of the `seed 
geodesics' $r(-\eps)$ of \eqref{eq:0-data} to $r(0)=r_0$. 
\end{proof}

\section{Limits}\label{sec:lim}
In this final section we investigate the limiting behaviour of the 
complete regularised geodesics $\gamma_\eps$ provided by 
Theorem~\ref{thm-global-ex} as the regularisation 
parameter $\eps$ goes to zero. Physically this amounts to calculating the 
geodesics of the impulsive full \emph{pp}-wave metric \eqref{ibm}. In fact 
this is only interesting if the geodesics are not parallel to or contained in 
the wave zone $u\in[0,L]$. So let again $\gamma=(x(s),s,r(s))$ be a `seed 
geodesic' starting `in front' of the wave zone with increasing~$s$. To 
simplify notations we will also briefly write $\gamma=(x,r)$ and denote the 
data  at the first impulsive surface by  $\gamma(0)=(x_0,r_0)$ and 
$\dot{\gamma}(0)=(\dot{x}_0,\dot{r}_0)$, respectively. Clearly motivated by the 
procedure 
which leads to the completeness result in Section \ref{sec:cil} we define the 
\emph{limiting geodesic} $\tilde{\gamma}$ by (see also Figure \ref{freg}, below)
\begin{align}\label{eq-gammalim}
 \tilde{\gamma}(s):=\begin{cases}
                     \gamma(s)\qquad &(s<0)\,,\\
		     \gamma^+(s)\qquad &(0\leq s< L)\,,\\
		     \gamma^{++}(s)\qquad &(s\geq L)\,,
                    \end{cases}
\end{align}
where $\gamma^+=(x^+,r^+)$ is a geodesic between the spikes, i.e., a 
solution of 
\begin{align}\label{eq:geoi}
D_{\dot{x}}\dot{x}^i = - 
\frac{1}{L}h^{ik}(a_{k,j}-a_{j,k})\dot{x}^j\,, \quad 
\ddot{r}&=\frac{1}{L}a_{(i||j)}\dot{x}^i\dot{x}^j - 
(h^{ik}\frac{1}{L}a_k(a_{i,j}-a_{j,i}))\dot{x}^j\,, 
\end{align}
with initial data (cf.\ \eqref{data1})
\begin{align}\label{ro-last}
 \gamma^+(0)&=(x_0^+,r_0^+):=\lim_{\eps\to 
0}\gamma_\eps(\eps)=\lim_{\eps\to0}(x_1^\eps,r_1^\eps)\,,\\ \label{ro-last1}
 \dot\gamma^+(0)&=(\dot x_0^+,\dot r_0^+):=\lim_{\eps\to 
0}\dot\gamma_\eps(\eps)=\lim_{\eps\to0}(\dot x_1^\eps,\dot r_1^\eps)\,,
\end{align}
where of course $\gamma_\eps$ is the global geodesic of Theorem 
\ref{thm-global-ex} associated with the `seed' $\gamma$. 
Furthermore, $\gamma^{++}=(x^{++},r^{++})$ is a background geodesic `behind'
the wave zone, i.e., $\gamma^{++}$ solves~\eqref{eq:bgrgeo} with initial 
data (cf.\ \eqref{eq:L+e-data})
\begin{align}
 \gamma^{++}(L) &= (x_0^{++},r_0^{++}):=\lim_{\eps\to0}{\gamma}(L+\eps)=
 \lim_{\eps\to0}(x^\eps_3,r^\eps_3)\,,\\
 \dot\gamma^{++}(L) &= 
(\dot x_0^{++},\dot 
r_0^{++}):=\lim_{\eps\to0}\dot{\gamma}(L+\eps)=\lim_{\eps\to0}
 (\dot x^\eps_3,\dot r^\eps_3)\,.
\end{align}

Now we turn to the explicit calculation of the limits of the data 
\eqref{ro-last} \eqref{ro-last1}, i.e., the behaviour of the 
limiting geodesic at the first spike.

\begin{prop}\label{lim-prop}
 Let $\gamma_\eps=(x_{\eps},r_\eps)$ given by Theorem \ref{thm-global-ex} with 
`seed' $\gamma$ as above. Then
 \begin{align}
    x_0^+ &= x_0\,,\qquad \big(\dot{x}_0^+\big)^i 
  = \big(\dot{x}_0\big)^i -\frac{1}{2}h^{ik}(x_0)\big(\frac{2 a_k(x_0)}{L} - 
\alpha 
H_{,k}(x_0)\big)\,,\\
r_0^+ 
&= r_0 + \frac{\alpha}{2}H(x_0)\,,\\
  \dot{r}_0^+ 
&= \dot{r}_0 + \alpha 
H_{,j}(x_0)\Big(\frac{\dot{x}^j_0}{2}-\frac{1}{8}h^{jk}(x_0)\big(\frac{2 
a_k(x_0)}{L} - \alpha 
H_{,k}(x_0)\big)+h^{jk}(x_0)a_k(x_0)\Big)\nonumber\\
 &\quad + \frac{1}{2L^2}h^{jk}(x_0)a_k(x_0)a_j(x_0)\,.
 \end{align}
\end{prop}

\begin{proof}
We only sketch these overly technical calculations. First, $x_0^+=x_0$ is a 
direct consequence of Lemma \ref{lem-lim-xeps}(i). To obtain $\dot{x}_0^+$ 
by Lemma \ref{lem-lim-xeps}(ii) we only have to read off $F_2(x_0)$ from 
\eqref{eq:regeoe}.

For $r_0^+$ and $\dot{r}_0^+$ we use the integral 
equation for $r_\eps$ with the the integral equation for $\dot{x}_\eps$
inserted, together with the identity
$\int_{-\eps}^\eps\delta_\eps(s)\int_{-\eps}^s\delta_\eps(\sigma)\d \sigma \d s 
= \frac{1}{2}$ and Lemma \ref{lem-lim-xeps}. We only detail this in case 
of $r_0$:
\begin{align*}
 r_\eps(\eps) =& r^\eps_0 + 2\eps \dot{r}^\eps_0 + \underbrace{\frac{1}{L}\int_{-\eps}^\eps\int_{-\eps}^s \theta_\eps 
Da\dot{x}_\eps\dot{x}_\eps\d \sigma \d s}_{\rom{I}} - \underbrace{\frac{1}{L^2} \int_{-\eps}^\eps\int_{-\eps}^s h a 
\theta_\eps^2 Da \dot{x}_\eps\d \sigma \d s}_{\rom{II}}\\
&+\underbrace{\int_{-\eps}^\eps\int_{-\eps}^s DH \delta^\eps_{\alpha,\beta}\dot{x}_\eps\d \sigma \d s}_{\rom{III}} - 
\underbrace{\frac{1}{L}\int_{-\eps}^\eps\int_{-\eps}^s h a^2 \theta_\eps\delta^\eps_{L^{-1},L^{-1}}\d \sigma\d s}_\rom{IV}\\
&+\underbrace{\frac{1}{2L}\int_{-\eps}^\eps\int_{-\eps}^s h a \theta_\eps DH \delta^\eps_{\alpha,\beta}\d \sigma\d 
s}_\rom{V} + \underbrace{\frac{1}{2}\int_{-\eps}^\eps\int_{-\eps}^s H (\delta^\eps_{\alpha,\beta})_{,u}\d \sigma \d 
s}_\rom{VI}\,, 
\end{align*}
where clearly $\rom{I}$ - $\rom{V}$ are $O(\eps)$. For example $|\rom{III}|\leq 2\eps \|DH\||\alpha|K (|\dot{x}_0| + c + 
K\|F_2\|_{\infty,I_1})$. Finally, we estimate $\rom{VI}$:
\begin{align*}
 |\rom{VI}-\frac{\alpha}{2}H(x_0)| &\leq \frac{1}{2} \int_{-\eps}^\eps\int_{-\eps}^s |H(x_\eps(\sigma)) - H(x_0)| 
|(\delta^\eps_{\alpha,\beta}(\sigma))_{,u}|\d \sigma\d s\\
&\leq \frac{1}{2}\sup_{\sigma\in[-\eps,\eps]}|H(x_\eps(\sigma))- 
H(x_0)| |\alpha| 4 \|\rho'\|\,,
\end{align*}
where we used that $\int_{-\eps}^\eps\int_{-\eps}^s\delta_\eps'(\sigma)\d \sigma\d s = 1$ and 
$\sup_{\sigma\in[-\eps,\eps]}|H(x_\eps(\sigma))- 
H(x_0)|$ converges to zero by Lemma \ref{lem-lim-xeps}.

The most difficult case is $\dot{r}_0^+$, since $(\dot r_\eps)_\eps$ is not uniformly bounded in $\eps$. However, 
$(\dot{r}_\eps(\eps))_\eps$ converges, as can seen as follows: As above write $\dot{r}_\eps(\eps)=\dot{r}^\eps_0 + \rom{I'}+ 
\rom{II'}+ \rom{III'}+ \rom{IV'}+ \rom{V'}+ \rom{VI'}$. Then $\rom{I'}+\rom{II'} = O(\eps)$ and in $\rom{III'}$ one inserts 
the integral equation for $\dot{x}_\eps$ and uses that $\int_{-\eps}^\eps\delta_\eps(s)\int_{-\eps}^s\delta_\eps(\sigma)\d 
\sigma \d s = \frac{1}{2}$. Finally, $\rom{IV'}$ and $\rom{V'}$ can be handled similarly and for $\rom{VI'}$ one uses 
integration by parts to obtain $\frac{\alpha}{2}\int_{-\eps}^\eps \delta_\eps' H \d s = 
-\frac{\alpha}{2}\int_{-\eps}^\eps\delta DH \dot{x}_\eps \d s$, which can be handled as $\rom{III'}$.
\end{proof}

We see from the explicit expressions given in Proposition \ref{lim-prop} that 
the limiting geodesic $\tilde \gamma$ displays the following 
behaviour at the first spike: The $x$-component is continuous with a finite 
jump in its velocity, since $\dot x_0^+\not=\dot x_0$ (in general). On the other hand, the 
$r$-component itself is discontinuous suffering a finite jump and the same is 
also true for its derivative. This behaviour is correlated with the fact that 
while $\dot r_\eps$ is \emph{not} uniformly bounded on $[-\eps,\eps]$ its value 
when leaving the regularisation strip $\dot r_\eps(\eps)$ \emph{is} 
nevertheless uniformly bounded in $\eps$.  
\medskip

Now we may analogously calculate the limits at the second spike to obtain
\begin{prop}
 Let $\gamma_\eps=(x_{\eps},r_\eps)$ given by Theorem \ref{thm-global-ex} with 
`seed' $\gamma$ as above. Then
 \begin{align}
    x_0^{++} 
   &= x^+(L)\,,\\
  (\dot{x}_0^{++})^i 
 &= (\dot{x}^+(L))^i  
-\frac{1}{2}h^{ik}(x^+(L))\big(\frac{2 
a_k(x^+(L))}{L} - \beta H_{,k}(x^+(L))\big)\,,\\
r_0^{++} 
&= r^+(L) + 
\frac{\beta}{2}H(x^+(L))\,,\\
  \dot{r}_0^{++} 
&= \dot{r}^+(L) \nonumber\\
 &\quad + \beta 
H_{,j}(x^+(L))\Big(\frac{(\dot{x}^+(L))^j}{2}-\frac{1}{8}h^{jk}
(x^+(L))\big(\frac{2 a_k(x^+(L))}{L} - \beta 
H_{,k}(x^+(L))\big)\\
&\qquad\qquad\qquad\quad+h^{jk}(x^+(L))a_k(x^+(L))\Big)\nonumber\\
&\quad + \frac{1}{2L^2}h^{jk}(x^+(L))a_k(x^+(L))a_j(x^+(L))\,.\nonumber
 \end{align}
\end{prop}

Finally we may prove the actual convergence result, saying that the regularised 
complete geodesics $\gamma_\eps$ of Theorem~\ref{thm-global-ex} converge to the 
limiting geodesics $\gamma$ of \eqref{eq-gammalim} consisting of appropriately 
matched geodesics of the complete background and the `intermediate' region. 

\begin{thm}\label{thm:lim}
Let $\gamma_\eps=(x_{\eps},r_\eps)$ be the complete geodesic of Theorem 
\ref{thm-global-ex} with `seed' $\gamma$ as above. Then $\gamma_\eps$ 
converges to the limiting geodesic $\tilde\gamma=(\tilde x,\tilde r)$ of 
\eqref{eq-gammalim} in
the following sense:
\begin{enumerate}
  \item $x_{\eps}\to \tilde x$ locally uniformly, \\
  $\dot x_{\eps}\to\dot{\tilde{x}}$ as distribution and\\
  uniformly on compact intervals not containing $t=0$ or $t=L$.
  \item $r_\eps\to \tilde r$ as distribution and \\
  in $\Con^1$ on compact intervals not containing $t=0$ or $t=L$.
\end{enumerate}
\end{thm}

Observe that the notions of convergence in the theorem are optimal given the 
regularity of the limits: $\dot{\tilde x}$ is discontinuous at $s=0$ and 
$s=L$ and so uniform convergence can only hold on bounded intervals not 
containing these two points. The same reasoning applies to $\dot{\tilde r}$ and 
its derivative.

\begin{figure}[htb]
\caption{The limiting behaviour of the complete geodesics $\gamma_\eps$ with 
two values of the regularisation parameter $0<\eps_2<\eps_1$ exemplified. The 
$x$-components of the seed and limiting geodesics are drawn in bold black, the 
regularised geodesics given by Proposition \ref{ex} are drawn dotted in 
green ($x_{\eps_1}$) and blue ($x_{\eps_2}$), respectively. Note that 
$x_{\eps}(\eps)$ converges to $x(0)$ for 
$\eps\searrow0$ and similarly, ${x}_\eps(L+\eps)\to x^+(L)$.}
\begin{center}\label{freg}
\includegraphics[scale=1.1]{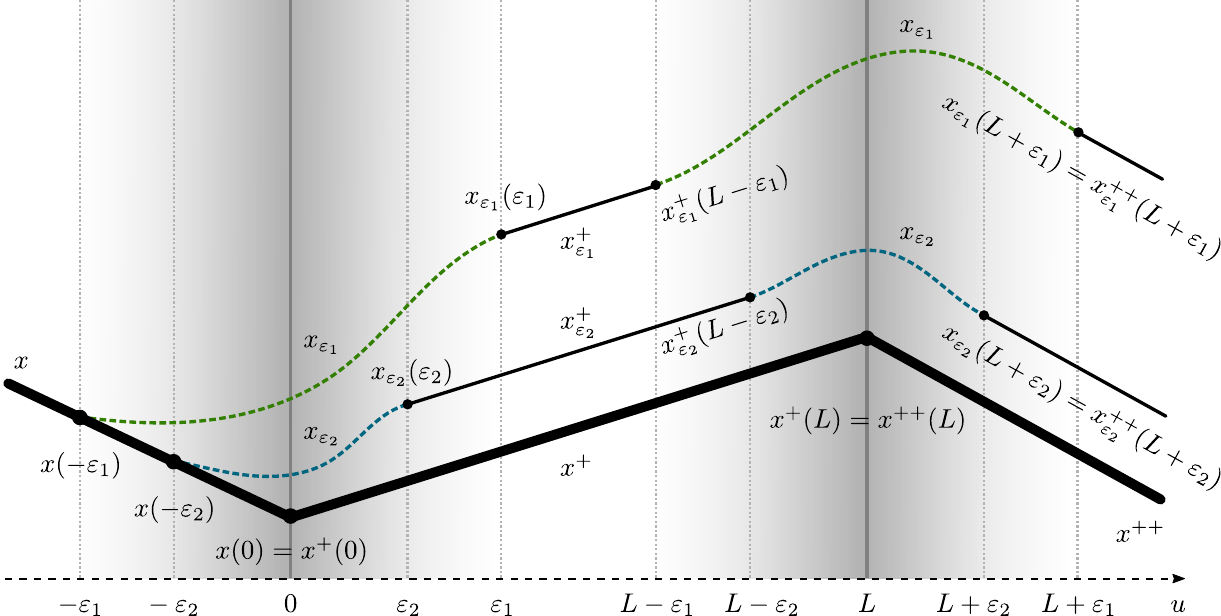}
\end{center}
\end{figure}

\begin{proof}
 Let $T>0$ then on $[-T,-\eps]$ the regularised geodesic $\gamma_\eps$ is 
equal to the `seed geodesic' $\gamma$ since both solve the same initial value problem.

Now on $[\eps,T]$ with $T\leq L-\eps$ (the first spike), $\gamma_\eps$ and $\gamma^+$ solve the 
same ODE, i.e., 
\eqref{eq:geoi} but with different initial conditions. By 
continuous dependence of the solutions of ODEs on initial data 
we have for $t\in[\eps,T]$
\begin{align}
\max(|\gamma_\eps(t)-\gamma^+(t)|,|\dot{\gamma}_\eps(t)-\dot{\gamma}^+(t)|)\leq 
\max(|\gamma_\eps(\eps)-\gamma^+(\eps)|,|\dot{\gamma}_\eps(\eps)-\dot{\gamma}
^+(\eps)|)e^{LT}\,,
\end{align}
where $L$ is a Lipschitz constant for the ($\eps$-independent) right-hand-side 
of \eqref{eq:geoi} (on some suitable bounded set). Now we can 
insert $\gamma^+(0)$ to obtain
\begin{align}
 |\gamma_\eps(\eps)-\gamma^+(\eps)|\leq |\gamma_\eps(\eps)-\gamma^+(0)| + 
|\gamma^+(0)-\gamma^+(\eps)|\to 0\,,
\end{align}
since $\gamma^+(0)=\lim_{\eps\searrow 0}\gamma_\eps(\eps)$ and $\gamma^+$ is continuous. Analogously, we insert 
$\dot{\gamma}^+(0)$ to obtain
\begin{align}
 |\dot{\gamma}_\eps(\eps)-\dot{\gamma}^+(\eps)|\leq |\dot{\gamma}_\eps(\eps) - 
\dot{\gamma}^+(0)| + 
|\dot{\gamma}^+(0)-\dot{\gamma}^+(\eps)|\to 0\,,
\end{align}
since $\dot{\gamma}^+(0)=\lim_{\eps\searrow 0}\dot{\gamma}_\eps(\eps)$ and 
by the continuity of $\dot{\gamma}^+$.
This gives uniform convergence of $(\gamma_\eps)_\eps$ on any compact interval 
not containing $t=0$ (in its interior).

On $[-\eps,\eps]$ Lemma \ref{lem-lim-xeps} yields the convergence of 
$(x_\eps)_\eps$ to $x_0$ and to establish the global 
distributional convergence of $(r_\eps)_\eps$ it remains only to consider 
$(r_\eps)_\eps$ on $[-\eps,\eps]$. So let 
$\phi\in\D(\R)$ and by Lemma \ref{lem-v-bound} there is a $C>0$ such that 
$|r_\eps(t)|\leq C$ and thus
\begin{align}
 |\int_{-\eps}^\eps (r_\eps(s)-\tilde{r}(s))\phi(s)\d s| \leq 2\eps C 
\|\phi\|_\infty \to 0\,.
\end{align}

Finally, the second spike, i.e., $[L-\eps,L+\eps]$, and behind the wave zone, i.e., $(L+\eps,\infty)$, can be handled 
analogously.
\end{proof}

\section{Conclusion}
In this contribution we have provided completeness results both for the 
extended as well as for the impulsive case of full \emph{pp}-waves.
This class of geometries allows for an arbitrary $n$-dimensional 
Riemannian manifold $N$ as a wave surface and for non-trivial 
off-diagonal terms in the metric (encoding the internal spin of the source), 
hence includes as special cases classical \emph{pp}-waves, N-fronted waves with 
parallel rays, and gyratons alike. In the extended case we have generalised 
the results on NPWs by providing a sufficient criterion for completeness in 
terms of the spatial asymptotics of the metric functions with a certain (local) 
uniformity with respect to proper time. In the impulsive case we have employed a 
regularisation approach to prove that \emph{all} these geometries are complete 
(provided the spatial profile functions are smooth). This confirms earlier 
results saying that the effect of the spatial asymptotics of the metric 
functions on completeness is wiped out in the impulsive limit. Finally we have 
explicitly derived the geodesics in the impulsive case in terms of a matching 
of corresponding background geodesics. This result, in particular, allows to derive the particle motion in
the field of specific ultrarelativistic particles possessing an internal spin opening the road to applications in quantum 
scattering and high energy physics.

\section*{Acknowledgement}
We thank Ji\v{r}\'i Podolsk\'y for numerous discussions and for generously 
sharing his experience. R.\v{S}. was supported by the grants GA\v{C}R 
P203/12/0118, UNCE~204020/2012 and the Mobility grant of the Charles 
University. C.S. and R.S. were supported by FWF grants  P25326 and P28770.

\bibliographystyle{alphaabbr}
\bibliography{gc_RoSv}

\end{document}